%% file: main.tex
\documentclass[sigconf, nonacm]{acmart}

\usepackage{times,epsfig}

\usepackage{booktabs} 
\usepackage{color,soul}

\usepackage{fancyvrb}
\usepackage{amsmath}
\usepackage{amsthm}
\usepackage{float}
\usepackage[boxed,noend]{algorithm2e}
\usepackage{graphicx}
\usepackage{lipsum}
\usepackage{xfrac}

\usepackage{colortbl}
\usepackage{multirow}
\usepackage{times} 
\usepackage{fancyvrb}
\usepackage{balance}
\usepackage{colonequals}
\usepackage{paralist}
\usepackage{booktabs}
\usepackage{tabularx}

\usepackage{etoolbox}
\usepackage{caption}
\usepackage{subcaption}

\let\chapter\undefined

\SetKwInput{KwInput}{Input}
\SetKwInput{KwOutput}{Output}

\fvset{frame=single,framesep=1mm,fontfamily=tt,numbers=left,framerule=.3mm,numbersep=1mm,commandchars=\\\{\}}
\definecolor{bgred}{RGB}{255,210,205}
\definecolor{bgblue}{RGB}{210,220,255}
\definecolor{bgyellow}{RGB}{255,255,109}
\definecolor{bggrey}{RGB}{223,223,225}
\definecolor{bgpink}{RGB}{255,202,239}
\definecolor{bgsky}{RGB}{182,219, 255}
\definecolor{purple}{RGB}{180,0,180}

\definecolor{hvygreen}{RGB}{0,128,0}
\definecolor{hvypink}{RGB}{102,0,51}

\newcommand{\name}[1]{\operatorname{#1}}
\newcommand{\funct}[2]{\operatorname{#1}\!\left({#2}\right)}
\newcommand{\To}{\!\rightarrow\!}
\newcommand{\card}[1]{\left|{#1}\right|}

\renewcommand{\epsilon}{\varepsilon}
\renewcommand{\phi}{\varphi}

\newcommand{\oofi}[1]{\operatorname{O}({#1})}

\newcommand{\subtab}{\textcd{SubTab}}
\newcommand{\tsub}[1][T]{#1_{\textrm{sub}}}

\newcommand{\textcd}[1]{{\footnotesize \textup{\textsf{#1}}}}
\newcommand{\tcol}[1]{\textit{#1}}

\ifdef{\removecomments}
{
    \newcommand{\tova}[1]{}
    \newcommand{\amit}[1]{}
    \newcommand{\kathy}[1]{}
    \newcommand{\susan}[1]{}

}
{
    \newcommand{\annotation}[1]{[[[#1]]]}
    
    \newcommand{\kathy}[1]{\textbf{\small\textcolor{hvypink}{\annotation{(Kathy)~#1}}}{\typeout{#1}}}
    \newcommand{\tova}[1]{\textbf{\small\textcolor{purple}{\annotation{(Tova)~#1}}}{\typeout{#1}}}
    \newcommand{\amit}[1]{\textbf{\small\textcolor{blue}{\annotation{(Amit)~#1}}}{\typeout{#1}}}
    \newcommand{\yael}[1]{\textbf{\small\textcolor{hvygreen}{\annotation{(Yael)~#1}}}{\typeout{#1}}}
    \newcommand{\susan}[1]{\textbf{\small\textcolor{purple}{\annotation{(Susan)~#1}}}{\typeout{#1}}}

}
\newcommand{\eat}[1]{}

\newcommand\vldbpagestyle{plain} 
\begin{document}

\title{Selecting Sub-tables for Data Exploration}

\author{Kathy Razmadze}
\affiliation{%
  \institution{Tel Aviv University}
}
\email{kathyr@mail.tau.ac.il}

\author{Yael Amsterdamer}
\affiliation{%
  \institution{Bar-Ilan University}
  }
\email{amstery@cs.biu.ac.il}

\author{Amit Somech}
\affiliation{%
  \institution{Bar-Ilan University}
  }  
\email{somecha@cs.biu.ac.il}

\author{Susan B. Davidson}
\affiliation{%
  \institution{University of Pennsylvania}
  }
\email{susan@cis.upenn.edu}

\author{Tova Milo}
\affiliation{%
  \institution{Tel Aviv University}
  }
\email{milo@post.tau.ac.il}
\begin{abstract}
We present a framework for creating small, informative sub-tables of large data tables to facilitate the first step of data science: data exploration. Given a large data table table $T$, the goal is to create a sub-table of small, fixed dimensions, by selecting a subset of $T$'s rows and projecting them over a subset of $T$'s columns. The question is: which rows and columns should be selected to yield an informative sub-table? 

We formalize the notion informativeness based two complementary metrics: {\em cell coverage}, which measures how well the sub-table captures prominent association rules in $T$,  and {\em diversity}. Since computing optimal sub-tables using these metrics is shown to be infeasible, we give an efficient algorithm which indirectly accounts for association rules using table embedding. The resulting framework can be used for visualizing the complete sub-table, as well as for displaying the results of queries over the sub-table, enabling the user to quickly understand the results and determine subsequent queries. Experimental results show that we can efficiently compute high-quality sub-tables as measured by our metrics, as well as by feedback from user-studies.

\eat{
We start by defining a notion of informativeness based on prominent data patterns within and across columns. Briefly, we use the standard notion of association rules and define a metric based on the rules represented in the sub-table. This metric captures, intuitively, the fraction of the data represented or ``covered'' by the sub-table, and we combine it with a metric accounting for the sub-table contents diversity.  Next, we study the complexity of the optimizing our metrics and propose different solutions for sub-table selection. In particular, we show an effective algorithm based on table embedding in vector representation. The resulting framework can be used for visualizing the complete sub-table, as well as for displaying the results of queries over the sub-table, enabling the user to quickly understand the results and determine subsequent queries. We conduct an experimental study over realistic datasets and data exploration tasks.  The results indicate that we can efficiently compute high-quality sub-tables both by our intrinsic metrics as well as in terms of their usefulness in external tasks \yael{with real users}.
}

\end{abstract}

\maketitle
\pagestyle{\vldbpagestyle}

\input{sections/intro}

\input{sections/related}
\input{sections/metrics}

\input{sections/complexity}

\input{sections/algorithm}

\input{sections/experiments}

\input{sections/conc}

\newpage
\bibliographystyle{ACM-Reference-Format}
\bibliography{sigproc} 

\end{document}

%% file: sections/intro.tex
\section{Introduction}
\label{sec: intro}

Data exploration is an important first step in data analytics.  During this step, the analyst tries to understand an unfamiliar dataset and determine what part of the data is relevant to their task by displaying the table, looking at the table description, or visualizing column values.  They may also run simple exploratory queries over the dataset, using selection, projection, sorting and grouping.  
However, when displaying a large input table or query result only a small subset of the table is typically shown -- and without input from the user, the subset is arbitrary. For example, the default display of Pandas\footnote{\textcd{Pandas}: Python Data Analysis Library. \url{https://pandas.pydata.org/}} tables using the Python \textcd{display()} command includes the first and last several rows and columns. Frequently, this is not very informative as the sub-table may contain a lot of missing values and/or fail to capture the range of possible values in a column;  it may also elide columns that are important for further exploration and analysis.

\begin{figure}
    \includegraphics[width=1.0\columnwidth]{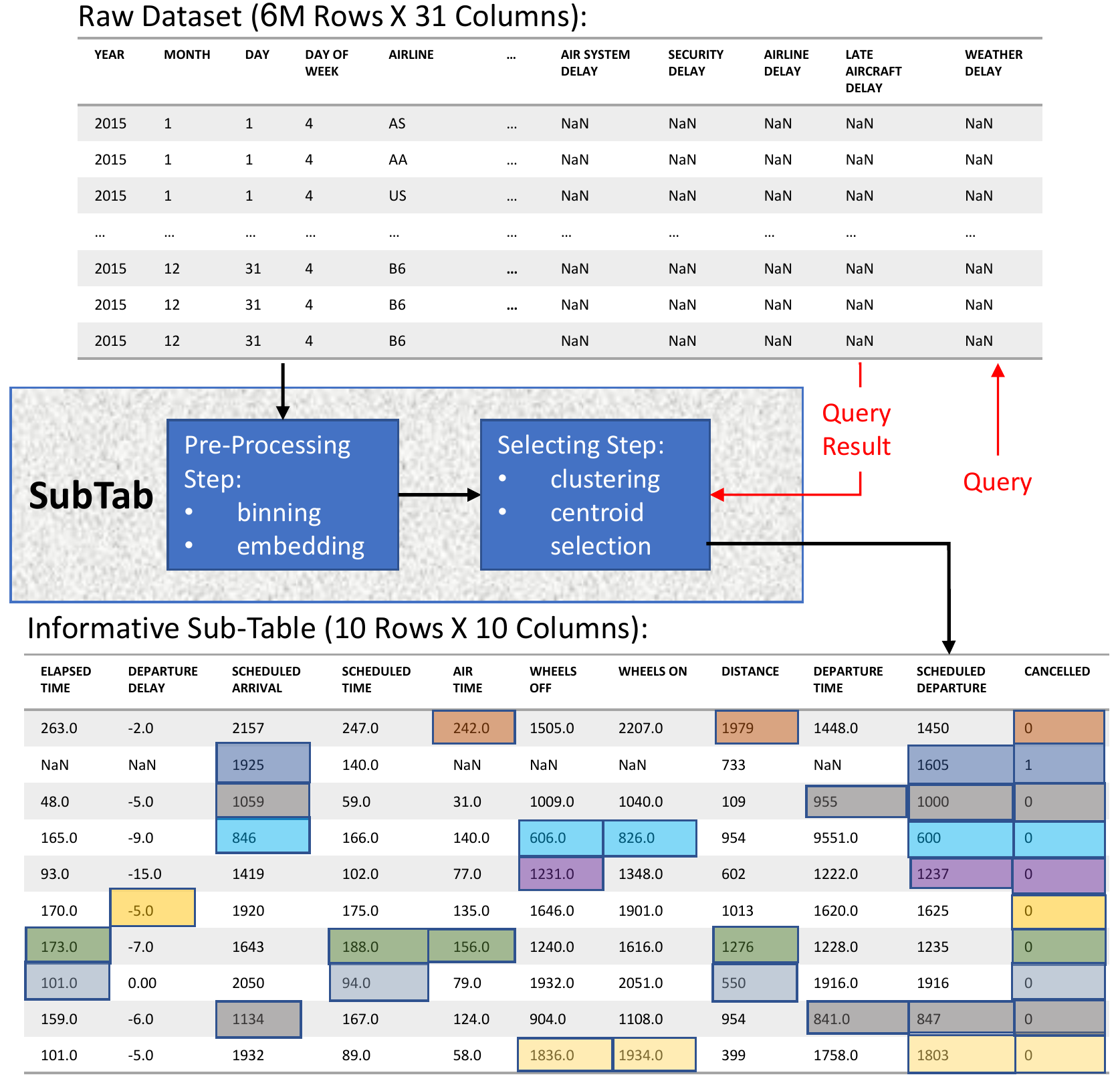}
    \vspace{-5mm}
    \caption{System Architecture}
    \vspace{-3mm}
    \label{fig:arch}
\end{figure}

\begin{example}\label{bad_example}
Consider a table $T$ taken from the Kaggle \emph{flights} dataset\footnote{https://www.kaggle.com/usdot/flight-delays?select=flights.csv} which contains 31 columns and $\sim$6M rows.  The analyst is using $T$ to predict flight cancellations, and hence is  interested in a specific target column, \tcol{CANCELLED}. The analyst starts by visually inspecting the data using Pandas \textcd{display($T$)}, which yields the table displayed at the top of Figure~\ref{fig:arch}.
This display of $T$ is not informative for the analysis task, as it does not include the target column. More importantly,  its usefulness for data exploration is limited: e.g., the last five columns contain only \textcd{NaN} values, and other columns include many repetitions of arbitrary values. \qed{}
\end{example}

Our goal is to support the data exploration task by selecting  small, informative sub-tables through which analysts can view data.
That is, given a table $T$ with $n$ rows and $m$ columns, our goal is to create a sub-table $\tsub$ with $k<<n$ rows and $l<<m$ columns which is a subset of $k$ rows  projected over a subset of $l$ columns of $T$. The sub-table should be \emph{informative} in the sense that it captures important data patterns within and across columns in~$T$, where we define patterns using the standard notion of \emph{association rules}.  The sub-table should also contain \emph{diverse} cell values out of the ones actually occurring in each selected column.
If one or more \emph{target columns} are known in advance to be the focus of the analysis, they must also be included in the~$l$ selected columns. 
 
\begin{example}\label{good_example}
 Continuing with the \emph{flights} dataset, an \emph{informative} sub-table is shown at the bottom of Figure~\ref{fig:arch}. 
 The sub-table captures several prominent association rules that hold over the input table and that include the target column  \tcol{CANCELLED}. We visualize this by highlighting, in each row, the cells that participate in a rule that holds for this row (many more rules hold, to avoid visual clutter we only highlight one rule per row). 
 For example, the first row of the sub-table exemplifies an association rule  stating that long flights ($\tcol{AIR\_TIME}\in[198.0, 422.0]$ and $\tcol{DISTANCE}\in[1546.0, 2724.0]$) are likely \emph{not} to be cancelled (highlighted in orange).
 The second row of the sub-table  exemplifies an association rule stating that short afternoon flights (according to the \tcol{SCHEDULED\_DEPARTURE}  and \tcol{SCHEDULED\_ARRIVAL} columns) are likely to be cancelled (\mbox{$\tcol{CANCELLED}=1$}, highlighted in blue).
Beyond these association rules, the sub-table gives useful insights about the column values, by showing diverse rows and diverse values per column. \qed{} 
\end{example}

As alluded to above, we formalize the notion of informativeness based on a combination of two complementary metrics: \emph{cell coverage} and \emph{diversity}. Cell coverage measures how well the sub-table $\tsub$ covers prominent\footnote{There are standard metrics we can use to measure the prominence of association rules in $T$, such as Support and Confidence~\cite{agrawal1994fast}.} association rules that hold over the input table $T$. 
Given a set of prominent association rules, we consider the rules that are captured by the sub-table as well as the marginal contribution of each rule, and combine them in a metric which reflects the number of cells in $T$ that are describable by association rules that are captured in $\tsub$. 
If one or more \emph{target columns} are known in advance to be the focus of the analysis, they will be included in the $l$ selected columns, and we measure cell coverage only according to association rules that include one or more target columns. 

For the second metric, \emph{diversity}, we rely on the average pair-wise similarity between the rows of the sub-table, using a Jaccard-like similarity measure that accounts for categorical as well as continuous columns. These metrics are combined in a \emph{score} to measure the informativeness of a sub-table.

Unfortunately, we  show that optimizing this informativeness score directly is infeasible. 
Moreover, even computing the association rules may not be practical in our setting: although there are several efficient techniques for mining association rules (e.g.,~\cite{AssociationRulesTables, agrawal1994fast,gunopulos2003discovering, omiecinski2003alternative, srikant1997mining}), they may still be overly time consuming for large datasets in an interactive setting.

We therefore consider a sub-table computation method which indirectly accounts for association rules using \textit{table embedding}~\cite{bordawekar2019exploiting,embeddings_integration_tasks,Cappuzzo_2020}. Given a table~$T$ we use \emph{binning}~\cite{AssociationRulesTables} to split each column's values into a  small set of meaningful groups. We then compute an \emph{embedding} of table cells as vectors. Several different methods have been recently proposed for this task~\cite{bordawekar2019exploiting,embeddings_integration_tasks,deng2020turl, tang2020rpt, Table2Vec2019,Cappuzzo_2020}, 
among which we chose a fast and effective embedding method based on Word2Vec~\cite{mikolov2013distributed} (and compare with other options in Section~\ref{sec: exp}). 

The embedding captures bin co-occurrences, and therefore implicitly corresponds to frequent itemsets and association rules.
To select rows and columns for a sub-table we derive from the cell vectors a vector representation for rows and columns, cluster them (separately) and select the centroids as rows and columns that represent diverse patterns in the data. 

Empirically, this method achieves near-optimal scores when  compared to upper bounds on  
coverage and diversity.

An important benefit of our solution design is in \emph{responding to queries over $T$}: during the exploratory data analysis (EDA) session, users typically issue different queries on a given table~$T$ (red arrows of Figure~\ref{fig:arch}). Our computation of cell embedding may be viewed as a part of the pre-processing step of a given data table~$T$, along with the binning of its values (first blue box in Figure~\ref{fig:arch}). This step  only has to be executed once upon loading the table. Then, if the analyst issues a selection-projection (SP) query on $T$ and wishes to view its result as a sub-table, we need only to compute the vector representation of rows and columns in $Q(T)$ based on the cells of $T$ that appear in them, and re-execute clustering and centroid selection (\emph{Selecting} step in Figure~\ref{fig:arch}, shown as the second blue box). This significantly speeds up sub-table computation compared with computing everything from scratch (a few seconds instead of up to a minute for large tables).

\begin{figure*}

    \includegraphics[width=\textwidth]{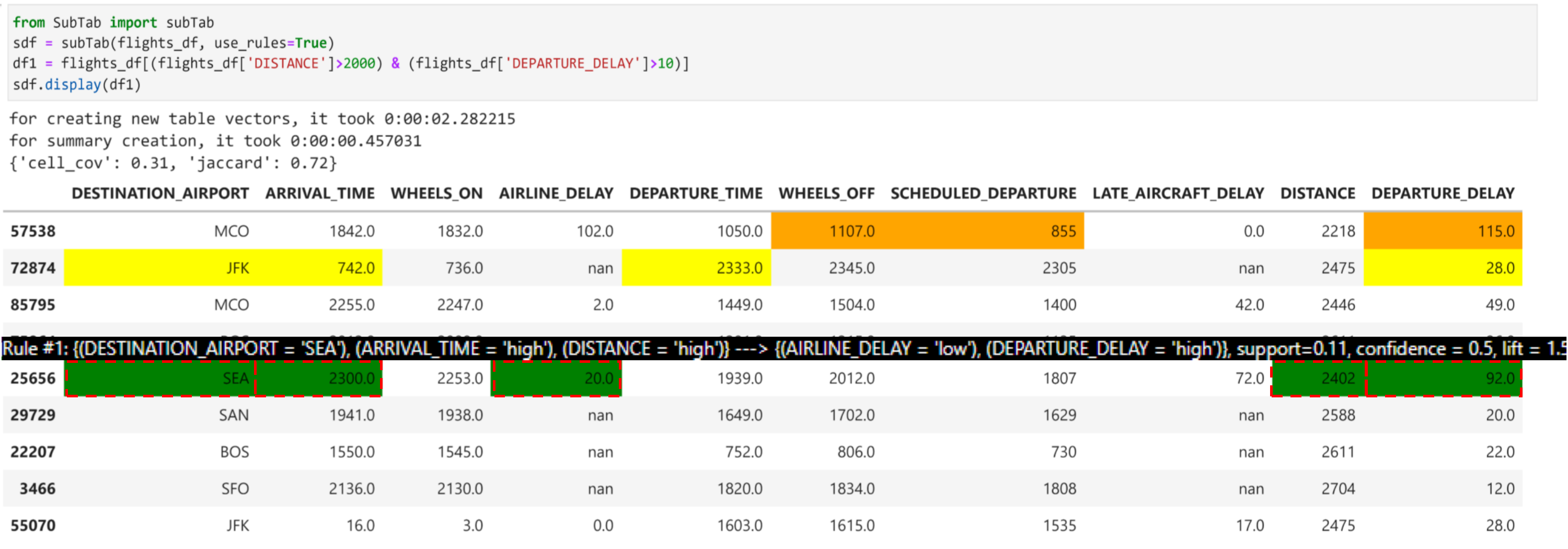}
    
\vspace{-2mm}
    \caption{Example usage of \subtab{} -- presenting an informative, 10X10 sub-table for a large query results-set}
    \label{fig:explore_3}
\end{figure*}

\paragraph*{Contributions.} The contributions of this work include:
\begin{enumerate}
\item
A notion of \emph{informativeness} for sub-tables, which includes both cell coverage and diversity.  Cell coverage measures how well the sub-table captures prominent association rules in the input table, while diversity ensures that repetition of values within the sub-table are minimized. 
\item
A {\em formalization and complexity analysis} of the problem of selecting an optimal sub-table,  which shows that computing optimal sub-tables using the informativeness metric is infeasible. 

\item A {\em greedy algorithm} for row selection which has approximation guarantees for optimal row selection, and a {\em semi-greedy algorithm} which further samples column combinations in order to find sub-tables with high cell coverage. Although the algorithm is not practical in an interactive setting, it serves as a baseline against which we compare the quality of computed sub-tables.
\item
A {\em practical algorithm}, \subtab, for computing informative sub-tables which accounts for association rules indirectly using table embedding.  The algorithm has two phases:  a pre-processing phase which performs binning and embedding, and can be executed as soon as the data is loaded; and a clustering and centroid selection phase that is called for each sub-table display, e.g., on query results.
\item An {\em implementation} for  \subtab{} 
as a local Python library that hooks into Pandas, and displays tables and query results as informative sub-tables. The implementation includes a UI which optionally highlights association rules (as in Figure~\ref{fig:arch}).
\item
{\em Experimental results} that measure the running time of \subtab{} and several baseline algorithms, as well as the quality of sub-tables produced.  Results show that the quality of sub-tables computed by \subtab{} exceed those of other interactive algorithms, and are comparable to algorithms that directly optimize our metric. 
Experiments with EDA sessions and user-studies show \subtab's effectiveness for enabling analysts to derive useful insights from the data. An analysis of its running time shows that \subtab{} is suitable for EDA sessions, and that the pre-processing step enables interactive displays of query results during the session. 

\end{enumerate}

\paragraph*{Organization.} The rest of the paper is organized as follows.  Related work is discussed in Section~\ref{sec: related}.  In
Section~\ref{sec: metrics} we define our metrics of cell-coverage and diversity. In Section~\ref{sec:complexity} we show hardness results for the problem of optimal sub-table selection, and the greedy and semi-greedy algorithms. In Section~\ref{sec:sol} we present our sub-table computation method based on table embeddings.  In Section~\ref{sec: exp} we discuss experimental results.  We conclude in Section~\ref{sec:conc}.

%% file: sections/related.tex
\section{Related work}
\label{sec: related}

\eat{There are three main lines of related work: (1) \textit{row sampling}, and (2) \textit{feature selection}, both aim to reduce the number of \textit{either} rows or columns in the data, then we further survey (3) \textit{data summarization} works, in which the goal is to generate an informative yet compressed summary of the data by employing a series of aggregations. 
}
There are three main lines of related work: (1) \textit{row sampling}, (2) \textit{feature selection}, and (3) \textit{data summarization}.
Row sampling (resp., feature selection) aims to reduce the number of rows (resp., columns) in the data, while the goal of data summarization is to generate an informative yet compressed summary of the data using a series of aggregations. 

\paragraph*{Row Sampling}
The task of sampling or selecting representative \textit{rows} from a large dataset has been studied in previous work for several different use cases. For example, work in \textit{Approximate Query Processing (AQP)} suggests using stratified sampling~\cite{agarwal2013blinkdb} and dynamic sampling~\cite{agarwal2013blinkdb,babcock2003dynamic} in order to reduce the number of tuples and produce faster yet inexact query results. 
Row sampling is used for \textit{efficiently generating data visualizations}~\cite{Visualization_aware_sampling}, i.e., reducing the number of underlying data points while minimizing the error in the produced visualization. This is typically done by using a visualization-inspired loss function~\cite{Dynamic_reduction_visual}. 
Another use case is \textit{query result diversification}~\cite{query_result_divers,liu2009structured}, in which the goal is to select rows from the query result that are both \textit{relevant} and \textit{diverse}. This is often done with a greedy algorithm which finds local-optimum solutions, given ad-hoc definitions of relevance and diversity~\cite{query_result_divers}.

In contrast, our approach \subtab{} generates a sub-table by directly selecting both rows \textit{and} columns. Also, while most of the works mentioned above are designed for a particular use case or have prior assumptions (e.g., a notion of relevance, or a specific query load~\cite{apq_kraska_2017}), \subtab{} captures generic patterns in the data, without using prior assumptions, for the purpose of producing informative, small samples of large tables. The latter allows \subtab{} to support many different data exploration use cases, information needs, and datasets.

\paragraph*{Feature Selection}
The task of reducing the number of columns in a dataset 
is an important step in many machine learning processes. The goal is typically to reduce the number of input variables considered by the ML model, which reduces the model's complexity as well as the training time~\cite{chandrashekar2014survey,info_entropy_feature_ranking}.
There is a plethora of work on feature selection (see~\cite{chandrashekar2014survey} for a survey), which can be roughly categorized as filter methods, that output the Top-k features w.r.t. a given metric (e.g., Chi-Square, ANOVA and Information-Gain)~\cite{info_entropy_feature_ranking}; as well as embedded~\cite{guyon2003introduction} and wrapper~\cite{john1994irrelevant} methods, which directly utilize the ML model to determine feature importance~\cite{cichocki2014era}.

While this work describes fundamental ML tools, they are ill-suited to our problem: first of all, they only select columns, and cannot be easily adapted to also select representative rows. 
Second, feature selection tools operate w.r.t a predefined, target column and a prediction task, which may not exist in the data exploration phase. 
As mentioned above, \subtab{} produces a sub-table by selecting both rows and columns, and does not require a target feature (although it will use target features if they are known). 

\paragraph*{Data summarization}
Another line of work attempts to derive compressed forms of the data, or produce a high-level, compact summary of the dataset.
These includes dimensionality reduction~\cite{cunningham2015linear} techniques, 
data sketches~\cite{cormode2017data} for online streams, and techniques for aggregation-focused AQP~\cite{hellerstein1999interactive,zeng2015g}.

Such summaries do preserve some properties of the original data, however they produce an \textit{altered}, compressed version. While these summaries are very useful for use cases such as AQP and feature engineering~\cite{cunningham2015linear}, they are not suitable for interactive analysis, where the user wants to see the actual data. In contrast, \subtab{} is focused on interactive data exploration, and thus efficiently generates sub-tables which contain representative rows and columns taken from the original table.

%% file: sections/metrics.tex
\section{Metrics}
\label{sec: metrics}

In this section, we define our metrics of cell coverage and diversity, and the problem of finding an optimal sub-table based on their combination. 
We then show that directly optimizing our metrics is infeasible in Section~\ref{sec:complexity}, and propose a practical solution in Section~\ref{sec: alg}.

\subsection{Model}
Sub-tables are formally defined as follows: A relational schema $U=\{u_1,\dots,u_{\card{U}}\}$ is a finite set of columns, such that each column $u_i$ allows values from a subset of the global domain $\mathcal{D}_i\subseteq\mathcal{D}$ (e.g., for a binary column, $\mathcal{D}_i=\{0,1\}$). A relational table over $U$  is a finite set $T\subseteq U\To\mathcal{D}$ of tuples such that $t(u_i)\in\mathcal{D}_i$ is the value of the cell in the row corresponding to a tuple $t\in T$ and the column $u_i\in U$.

\begin{definition}[Sub-table]\label{def:subtable}
Given a table $T$ over schema $U$, a sub-table $\tsub$ is a table over some schema $\tsub[U]\subseteq U$  such that each tuple $t\in \tsub$ is the projection of some tuple $t'\in T$ over the columns of $\tsub[U]$, i.e., for every $u\in\tsub[U]$, $t(u)=t'(u)$.
\end{definition}

We next define two standard notions that will be useful in the sequel: binning and association rules.

In a schema $U$, each column $u_i$ may be \emph{categorical}, namely,  $\mathcal{D}_i$ is discrete, e.g., a column of airline names; or \emph{continuous}, namely,  $\mathcal{D}_i$ is a continuous range, e.g., a column of flight distance. Moreover, in a table $T$ over $U$ a different distribution of values (e.g. uniform or skewed) may occur in each column. \emph{Binning} the column values is technique commonly used to allow a uniform treatment of columns with different ranges and distribution. Formally,

\begin{definition}[Binning]\label{def:bins}
Given a table $T$ over schema $U$, a \emph{binning} function $\mathcal{B}$ maps each column $u_i\in U$ to a finite set of bins $\mathcal{B}_i=\{B_1^i,\dots,B_{\card{\mathcal{B}_i}}^i\}$ such that for every $t\in T$, $t(u_i)$ belongs to exactly one bin $B_j^i\in \mathcal{B}_i$.
\end{definition}

We discuss in Section~\ref{sec: alg} the method we use for computing such a binning function.

\begin{example}\label{ex:binning}
In the flights dataset, the \tcol{DISTANCE} column is continuous, and so we split its range into the bins short, medium and long-distance. Depending on the column value distribution, we may obtain $B^{\tcol{DIST}}_{\textrm{long}}=[1546.0, 2724.0]$. The continuous  \tcol{AIRTIME} column may also have bins with the same labels, but different ranges matching is value distribution, e.g.,  $B^{\tcol{AT}}_{\textrm{long}}=[198, 422]$.  The \tcol{CANCELLED} column is binary, hence we can use its categories as bins. The \tcol{AIRLINE} column is also categorical but has many categories. We can create a smaller number of groups by e.g.\ splitting the  airlines according to the continent in which they are headquartered.
\end{example}

Next, we recall the notion of \emph{association rules}, which we use to capture patterns in the data.  Association rules will be used to measure and compare the quality of sub-tables.
Formally, 
\begin{definition}[Association rules~\cite{AssociationRulesTables}]
Given a table $T$ over schema $U$, an association rule $R$ has the form $\{(u_1,v_1),\dots(u_r,v_{r_R})\}\To\{(u_{{r_R}+1},v_{{r_R}+1}),\dots(u_{{r_R}+{p_R}},v_{{r_R}+{p_R}})\}$
where each
$u_i\in U$ is a column and each $v_i \in \mathcal{D}_i$ is a cell value. Denote by $U_R=\{u_1,\dots,u_{{r_R}+{p_R}}\}\subseteq U$ the set of columns used in $R$.
We say $R$ \emph{holds} for a tuple $t\in T$ if $t(u_i)=v_i$ for every $1\leq i \leq {r_R}+{p_R}$. Denote by $T_R\subseteq T$ the subset of tuples  for which $R$ holds.
\end{definition}  
Previous work includes different metrics for the quality of association rules, as well as corresponding algorithms for association rule \emph{mining}, i.e., the discovery of all association rules that meet some quality criteria (e.g.,~\cite{agrawal1994fast,gunopulos2003discovering,omiecinski2003alternative,srikant1997mining}).

In tables with diverse columns, the use of binning may improve the mined association rules~\cite{AssociationRulesTables}: given a table $T$, we can replace each cell value $t(u_i)$ with an identifier of its matching bin $B_j^i$. The resulting table would have a smaller number of distinct values per column, and each value would occur more frequently; consequently, one may be able to mine association rules that apply to many more tuples. 

\begin{example}\label{ex:rules}
Using the bins from example~\ref{ex:binning}, the association rule from Example~\ref{good_example} stating that long flights are likely \emph{not} to be cancelled can be written as: \\
\tcol{AIR\_TIME}$\in B^{\tcol{AT}}_{\textrm{long}}$,\tcol{DISTANCE}$\in B^{\tcol{DIST}}_{\textrm{long}}\To$\tcol{CANCELLED}$\in B^{\tcol{CANC}}_0$.
Similarly, the rule for flights likely to be cancelled can be written as \tcol{SCHEDULED\_DEPARTURE}$\in B^{\tcol{SD}}_{\textrm{afternoon}}$, \tcol{SCHEDULED\_ARRIVAL}$\in B^{\tcol{SA}}_{\textrm{afternoon}}\To$\tcol{CANCELLED}$\in B^{\tcol{CANC}}_1$.
\end{example}

\subsection{Informativeness Metrics}
\label{sec:cov}
We now develop quality metrics for sub-tables. The first type of metric that we develop intiutively measures how well data patterns in the full table are captured by the sub-table. 

\paragraph*{Cell coverage}
Given a sub-table $\tsub$ of table $T$ and a set a set $\mathcal{R}$ of association rules mined from $T$ (e.g., using~\cite{AssociationRulesTables}), to measure the \emph{coverage} of $\tsub$ with respect to $T$ and $\mathcal{R}$ we consider the following.
 \begin{enumerate}[(q1)]
     \item Which of the rules of $\mathcal{R}$ are \emph{covered}, i.e., captured by $\tsub$?
     \item What is the marginal contribution of each covered rule to $\tsub$'s informativeness?
     \item How do marginal contributions aggregate to an overall numerical score for $\tsub$?
\end{enumerate}

Since sub-tables include a subset of the table cells, and association rules are also defined at the level of table cells, we propose below formal definitions for q1-q3 that yield a \emph{cell coverage} metric. This metric intuitively reflects the ratio of cells in $T$ that are describable by association rules in $\mathcal{R}$ that are represented in $\tsub$. 

\begin{definition}[Cell coverage]\label{def:cov}
Let $T$ be a table, $\mathcal{R}$ a set of association rules mined from $T$, and $\tsub$ a sub-table of $T$. 
\begin{enumerate}[(d1)]
    \item A rule $R\in\mathcal{R}$ is said to be \emph{covered} by $\tsub$ if all the attributes of $R$ are in $\tsub$ ($U_R\subseteq \tsub[U]$), and there exists a tuple $t \in T_{sub}$ for which $R$ holds ($\{\tsub\}_R\neq \emptyset$). \\ Let $\tsub[\mathcal{R}]$ be the set of all rules in $\mathcal{R}$ that are covered by $\tsub$.
    \item The \emph{marginal contribution} of a rule $R\in\tsub[\mathcal{R}]$ is the subset of table cells it describes: \\
    $\funct{cell}{R,T}\colonequals \{\langle t, u\rangle\mid t\in T_R\wedge u\in U_R\}$.
    \item The \emph{cell coverage} of $\tsub$ with respect to $T, \mathcal{R}$ is denoted by
    \begin{align}\label{formula:coverage}
        \funct{cellCov_{\mathcal{R}}}{T,\tsub}\colonequals\frac{1}{\name{upcov}}\card{\bigcup_{R\in{\tsub[\mathcal{R}]}} \funct{cell}{R,T}}
    \end{align} I.e., it is the (normalized) number of cells in $T$ described by any covered rule in $R\in\tsub[\mathcal{R}]$. The normalization factor \(\name{upcov}\colonequals\card{\bigcup_{R\in\mathcal{R}}\funct{cell}{R,T}}\) is an upper bound on the number of cells that can be covered, and ensures that $\funct{cellCov_{\mathcal{R}}}{T,\tsub}\in[0,1]$. 
\end{enumerate}
\end{definition}

\begin{figure*}
\setlength{\tabcolsep}{3pt}
{\small
\begin{tabularx}{0.5\textwidth}{lllllll}
 \toprule $\hat{T}$ & \tcol{Row} &
\tcol{CANCELLED} & \tcol{DEP.\_TIME} & \tcol{YEAR} & \tcol{SCHED.\_DEP.} & \tcol{DISTANCE} \\
\midrule
&1& \cellcolor{bgpink}1& \cellcolor{bgpink}NaN & \cellcolor{bgpink}2015 & afternoon & \cellcolor{bgpink}short\\
&2&\cellcolor{bgsky}1& \cellcolor{bgsky}NaN & \cellcolor{bgsky}2015 &  \cellcolor{bgsky}afternoon & medium\\
&3&\cellcolor{bgyellow}1& \cellcolor{bgyellow}NaN & \cellcolor{bgyellow}2015 & morning & \cellcolor{bgyellow}medium\\
&4&\cellcolor{bgpink}1& \cellcolor{bgpink}NaN & \cellcolor{bgpink}2015 & \cellcolor{bgpink}morning & short\\
&5&\cellcolor{bgsky}0& \cellcolor{bgsky}morning & 2016 & \cellcolor{bgsky}morning & \cellcolor{bgsky}medium\\
&6&0& morning & 2015 & morning & medium\\
&7&\cellcolor{bgyellow}0& \cellcolor{bgyellow}evening & \cellcolor{bgyellow}2015 & evening & \cellcolor{bgyellow}long\\
&8&0& evening & 2015 & afternoon & long\\
\bottomrule
  \multicolumn{7}{l}{~}  \\[1mm]
\end{tabularx}\hspace{5mm}\begin{tabularx}{0.43\textwidth}{llllll}
\toprule $\tsub[\hat{T}^{(1)}]$ &\tcol{Row} &
\tcol{CANCELLED} & \tcol{DEP.\_TIME} & \tcol{YEAR} & \tcol{DISTANCE} \\
\midrule
&1& \cellcolor{bgpink}1& \cellcolor{bgpink}NaN & \cellcolor{bgpink}2015  & \cellcolor{bgpink}short\\
&5&\cellcolor{bgsky}0& \cellcolor{bgsky}morning & 2016 &  \cellcolor{bgsky}medium\\
&7&\cellcolor{bgyellow}0& \cellcolor{bgyellow}evening & \cellcolor{bgyellow}2015 &  \cellcolor{bgyellow}long\\
\bottomrule
\multicolumn{6}{l}{}  \\
\toprule $\tsub[\hat{T}^{(2)}]$ &\tcol{Row} &
\tcol{CANCELLED} & \tcol{DEP.\_TIME} & \tcol{YEAR} & \tcol{SCHED.\_DEP.}  \\
\midrule
&1& \cellcolor{bgpink}1& \cellcolor{bgpink}NaN & \cellcolor{bgpink}2015 & \cellcolor{bgpink}afternoon \\
&5&\cellcolor{bgsky}0& \cellcolor{bgsky}morning & 2016 & \cellcolor{bgsky}morning \\
&7&\cellcolor{bgyellow}0& \cellcolor{bgyellow}evening & \cellcolor{bgyellow}2015 & evening \\
\bottomrule
\end{tabularx}
}\caption{Example Table $\hat{T}$ with two sub-tables. {\normalfont Association rules (at most one per row) are highlighted.}}
        \label{fig:miniex}
\end{figure*}

In developing the cell coverage metric, we considered several alternative approaches. We next briefly describe some of these through an example, to illustrate the benefit of cell coverage.

\paragraph*{Alternative coverage metrics}
Consider the example table $\hat{T}$ on the left of Figure~\ref{fig:miniex}, illustrating some of the trends in the \emph{flights} dataset mentioned above. The table values represent bin names (e.g. ``short", ``medium", ``long"). Assume a set $\mathcal{R}$ of all association rules with column \tcol{CANCELLED} on the right, and at least two columns on the left, that hold for at least two rows. As an example, the rule \tcol{DEP.\_TIME}$=NaN$, \tcol{YEAR}$=2015\To$\tcol{CANCELLED}$=1$ applies to rows 1-4. For convenience, rules of maximal size are highlighted; each highlighted line illustrates a different rule, with colors alternating for clarity.

First, observe that rows with \tcol{CANCELLED}$=$1 are more similar to each other (have more values in common) compared with rows with \tcol{CANCELLED}$=$0 due to the fact that many fields do not apply (i.e.\ are NaN) when a flight is cancelled. 

As a result, there are~13 association rules for the first~4 rows, and only~8 for the last~4. This issue is exacerbated when there are many more columns with fixed values, leading to overlapping rules including subsets of these columns. Hence, using the number of rules covered as a measure of sub-table quality would lead to selecting many representatives from repetitive patterns, whereas we would like to see representatives from different areas of the data. Consequently, we propose to use measures based on \emph{data coverage} rather than rule coverage. 

Next, consider the two 
sub-tables on the right, $\tsub[\hat{T}^{(1)}]$ and $\tsub[\hat{T}^{(2)}]$, which differ only in the last attribute. These are ``good'' sub-tables, in the sense that they cover at least one rule for each tuple of $\hat{T}$. If we chose to use a row coverage metric, they would therefore have the same score. However,  $\tsub[\hat{T}^{(1)}]$ covers larger rules (two of size~4 and one of size~3) compared with $\tsub[\hat{T}^{(2)}]$ (two of size~3 and one of size~4). Accordingly, $\tsub[\hat{T}^{(1)}]$ describes~28 cells of $\hat{T}$, whereas $\tsub[\hat{T}^{(2)}]$ describes only~26. By normalizing these numbers, we would get cell coverage of~0.78 and~0.72, respectively, since~36 cells in total can be described by association rules. This motivates us to use a \emph{cell-based metric}. 

Finally, we note that, in $\tsub[\hat{T}^{(1)}]$, if we chose row~3 instead of~1 and row~6 instead of~5 we would have the same cell coverage. However, the sub-table would be more repetitive, containing only~2015 in the year field and two instances of medium distance. This demonstrates that coverage should be accompanied by a diversity metric, which we discuss next.

\paragraph*{Diversity} To generalize the example shown above, sub-tables with high cell coverage may seem repetitive to humans, since: 1) overlapping covered association rules may lead to repeating values in the sub-table, and 2) values that do not participate in rules are ignored. We therefore combine cell coverage with a diversity metric based on pairwise Jaccard similarity. As with cell coverage, binning the data values is useful in making our diversity metric more meaningful, and we consider two values in the same bin to be similar.

\begin{definition}[Diversity metric]\label{def:divers}
The \emph{similarity} of two tuples $t,t'\in \tsub$ is the ratio of cells which fall in the same bin, formally,
\[\funct{Jaccard}{t,t',\tsub}\colonequals\frac{\card{\{u_i\in\tsub[U]\mid \exists B_j^i\in\mathcal{B}(u_i).~ t(u_i),t'(u_i)\in B\}}}{\card{\tsub[U]}}\]
We then define the diversity of $\tsub$  as the complement of the average similarity between its tuples, namely,
\begin{align}\funct{divers}{\tsub, \mathcal{B}}\colonequals1 - \name{avg}_{t,t'\in\tsub}\funct{Jaccard}{t,t',\tsub}
\end{align}
\end{definition}

\begin{figure}
\setlength{\tabcolsep}{3pt}
{\small
\begin{tabularx}{\columnwidth}{llllll}
\toprule $\tsub[\hat{T}^{(3)}]$ &\tcol{Row} &
\tcol{CANCELLED} & \tcol{DEP.\_TIME} & \tcol{SCHED.\_DEP.} & \tcol{DISTANCE} \\
\midrule
&1& \cellcolor{bgpink}1& \cellcolor{bgpink}NaN  & afternoon & \cellcolor{bgpink}short\\
&5&\cellcolor{bgsky}0& \cellcolor{bgsky}morning & \cellcolor{bgsky}morning & \cellcolor{bgsky}medium\\
&7&\cellcolor{bgyellow}0& \cellcolor{bgyellow}evening & evening & \cellcolor{bgyellow}long\\
\bottomrule
\end{tabularx}}\caption{Example of a diverse sub-table}
        \label{fig:subtable3}
\end{figure}

\begin{example}
Consider again Table~$\tsub[\hat{T}^{(1)}]$ from Figure~\ref{fig:miniex}. Here the shown values are already bin names, so to compute the sub-table diversity we need to compute the ratio of identical cells in each pair of rows. In this case, the only overlaps are in \tcol{CANCELLED} in rows 5 and 7, and the \tcol{YEAR} of rows 1 and 7, yielding a diversity $\funct{divers}{\tsub[\hat{T}^{(1)}], \mathcal{B}}=1-\funct{avg}{0.25,0,0.25}= 0.83$. Figure~\ref{fig:subtable3} shows an even more diverse sub-table, with $\funct{divers}{\tsub[\hat{T}^{(3)}], \mathcal{B}}=1-\funct{avg}{0,0,0.25}= 0.92$, achieved by excluding the repetitive \tcol{YEAR} column. However, this table has a smaller cell coverage, describing only~24 cells, which indicates that there is a trade-off between cell coverage and diversity.
\end{example}

\paragraph*{Optimization problem}
Based on the metrics defined above, we define the optimization problem, \texttt{OPT-SUB-TABLE}. This problem computes a sub-table of a predefined size which balances high cell coverage and diversity. Moreover, we allow the user to specify target columns of interest  and focus the  sub-table on the target columns by including them in the chosen columns and by considering only association rules that include at least one of the target columns.

More formally, we are given as input a table~$T$ over schema $U$, dimensions~$k,l$ (the number of rows and columns, respectively), a set of target columns $U^*\subseteq U$ such that $\card{U^*}\leq l$,  a set of association rules~$\mathcal{R}$ mined from $T$, a binning $\mathcal{B}$  as in Def.~\ref{def:bins} and a parameter $\alpha\in[0,1]$ which is used to balance coverage and diversity (by default, $\alpha=0.5$). If there are target attributes ($U^*\neq\emptyset$), we retain only the rules the contain them:
$\mathcal{R}^*\colonequals\{R\in\mathcal{R}\mid \{u_{r_1},\dots, u_{{r_R}+{p_R}}\}\cap U^*\neq \emptyset\}$. Otherwise we retain all rules: $\mathcal{R}^*=\mathcal{R}$. Our goal is to find a $k\times l$ sub-table $\tsub$ that includes the target attributes ($U^*\subseteq \tsub[U]$), and that maximizes the following score among all such tables:
\begin{multline}
    \funct{combined}{\tsub,T,\mathcal{R}^*, \alpha} = \\
    \alpha\cdot \funct{cellCov_{{\mathcal{R}^*}}}{T,\tsub}+(1-\alpha)\cdot\funct{divers}{\tsub,\mathcal{B}}
\end{multline}
 
 \newcommand\mycommfont[1]{\rmfamily{#1}}
\SetCommentSty{mycommfont}
\begin{example}
Consider again sub-tables $\tsub[\hat{T}^{(1)}]$ and $\tsub[\hat{T}^{(3)}]$ in Figures~\ref{fig:miniex}-\ref{fig:subtable3}. Their combined scores for $\alpha=0.5$ are:\\
$0.5\cdot\sfrac{28}{36}+0.5\cdot0.83 = 0.80$ and \\
$0.5\cdot\sfrac{24}{36}+0.5\cdot 0.92=0.79$, respectively. \\
In fact, $\tsub[\hat{T}^{(1)}]$ is the optimal sub-table for this example.
\end{example}

Unfortunately, we will show in the next section that directly optimizing this problem is infeasible, and therefore give in Section~\ref{sec: alg} a practical solution that indirectly accounts for association rules using table embedding.

%% file: sections/complexity.tex
\section{Complexity}\label{sec:complexity}

We now analyze the complexity of \texttt{OPT-SUB-TABLE}. For simplicity, we will ignore the use of target columns and of binning, since in the extreme case the set of target columns may be empty and binning may assign each value to a separate bin. We mostly focus here on the sub-problem of \texttt{MAX-CELL-COVER}, namely, finding the sub-table with maximal cell coverage, i.e. solving \texttt{OPT-SUB-TABLE} with $\alpha=1$.  We denote by \texttt{DEC-CELL-COVER}  the corresponding decision problem: given a table~$T$ over schema $U$, dimensions~$k,l$, a set of association rules~$\mathcal{R}$ and a threshold $\Theta$, decide whether there exists a sub-table $\tsub$ of size  $k\times l$  whose cell coverage is $\funct{cellCov_{{\mathcal{R}}}}{T,\tsub}\geq \Theta$.

\subsection{Hardness}\label{sec:hardness}
Let $n$ and~$m$ be the numbers of tuples and columns, respectively, in the input table~$T$. A brute-force algorithm can theoretically traverse all $\oofi{n^k\cdot m^l}$ sub-tables of size $k\times l$ and find the one with maximal score. While this algorithm is polynomial in the size of $T$, it is infeasible due to the exponential dependency in $k,l$. For reasonable dataset and sub-table sizes, e.g., for $n=10,000$ and $m,k,l=5$, we need to check $10^{23}$ sub-tables. 

Since $k$ is small, a solution that is exponential in~$k$ may still be feasible, but not when the basis is~$n$. Thus, we examine whether our problem is \emph{fixed-parameter tractable}: FPT is class of fixed-parameter tractable problems, defined by having a solution in time $\oofi{\funct{poly}{n}\cdot f(k)}$ where $n$ is the input size, $f$ is a function and $k$ is a parameter. Unfortunately, we can show that our problem is probably not in FPT by the following proposition.

\begin{proposition}\label{prop:w1hardness}
Given a $n\times m$ table $T$ over schema $U$, and sub-table dimensions $k,l$.  
\texttt{Dec-Cell-Cover} is W[2]-hard with respect to $n=\card{T}$ and~$k$ as a parameter, assuming $m,l=\oofi{n}$.
\end{proposition}
\begin{proof}
\texttt{Dominating Set} is the problem of, given an un-directed graph of $n$ vertices, select $k$ vertices such that every vertex in the graph is either in the selected set or is connected by an edge to a vertex in the set. This problem is known to be W[2]-complete~\cite{downeyF1999}. The W-hierarchy is a hierarchy of classes defined by properties of the translation of problems into combinatorial circuits. It is known that FPT$=$W[0]$\subseteq$W[1]$\subseteq$W[2]$\subseteq\dots$, and conjectured that this hierarchy is strict, i.e., W[0]$\subset$W[1]$\subset$W[2]$\subset\dots$. We show a reduction from Dominating Set  to \texttt{Dec-Cell-Cover}, thus showing the latter is also W[2]-hard and hence not in FPT unless FPT=W[2].

\texttt{Dominating Set}$\leq$\texttt{Dec-Cell-Cover}. Given a graph $G=(V,E)$, define a table $T$ with a tuple $t_v$ for every $v\in V$, and an attribute $u_v$ for every $v\in V$, such that $t_v(a_u)=1$ iff $(u,v)\in E$ or $u=v$, and is NULL otherwise. Let $\mathcal{R}$ be composed of  association rules of the form $\{u_v\in\{1\}\}\To\{\}$ is in  for every $u_v$, and seek a summary of size $k\times \card{V}$, i.e., no column selection is required. In this case, $k$ tuples that cover all the non-NULL cells in~$I$ correspond to a dominating set of size~$k$, and vice versa. 
\end{proof}

The above proof assumes that $m,l=\oofi{n}$, i.e., the number of attributes is large. However, in practical cases it is often assumed that $m<<n$. In this case, our reduction does not apply, and in particular, if $m=\oofi{1}$, a dominating set with bounded degree is not W[2]-hard.  However, we can show NP-hardness in~$k$.

\begin{proposition}\label{prop:nphardness}
\texttt{Dec-Cell-Cover} is NP-hard in~$k$, the number of tuples selected for the summary,  even assuming the number of attributes $m=\oofi{1}$.
\end{proposition}
\begin{proof}
The proof is by a reduction from \texttt{Vertex Cover}, which is known to be hard in~$k$ for the selection of~$k$ vertices that are adjacent to every edge in a given graph, even if the maximal degree of a vertex is~3. Given a graph $G=(V,E)$ we assign each edge $e\in E$ a serial number $\funct{num}{e}$ and define an input table $I$ with~5 attributes and~$n$ tuples, such that for each edge $(u,v)\in E$ there exists tuples $t_u,t_v\in I$ and an attribute $a$ where $t_u(a)=t_v(a)=\funct{num}{(u,v)}$. Since each $v\in V$ appears in at most~3 edges, the other edges of $u$ and $v$ occupy at most~4 attributes of $t_u$ and $t_v$, and hence we can use the fifth attribute for $\funct{num}{(u,v)}$. Now, similarly to the proof of Prop.~\ref{prop:w1hardness}, we will construct the rules such that $\{a\in\{\funct{num}{(e}\}\}\To\{\}$ is a rule in $\mathcal{R}$, and so, the maximum coverage summary of size $k\times 5$ covers all the non-NULL cells iff the vertices corresponding to its tuples form a vertex cover in $G$.\hfill
\end{proof}

\paragraph*{Hardness of diversity optimization} We have so far focused on the hardness of cell coverage maximization (corresponding to solving \texttt{OPT-SUB-TABLE} with $\alpha=1$). However, we note that diversity maximization (corresponding to \texttt{OPT-SUB-TABLE} with $\alpha=0$) is also hard:  this problem (without column selection) was proven to be NP-hard using an analogous diversity definition in the context of selecting a diverse group of crowd workers~\cite{wu2016hear}, and the results hold for our setting as well. Since \texttt{OPT-SUB-TABLE} is strictly harder than both sub-problems, we can conclude that it is NP-hard.

\subsection{Approximate Solutions and Limitations}\label{sec:approx}
Given the hardness results above, we consider approximate solutions to \texttt{Max-Cell-Cover}, i.e., computing sub-tables with approximately-optimal score. 
In some cases, it may be feasible to enumerate all possible combinations of selecting~$l$ attributes for the summary, e.g., when $m=\oofi{1}$ or~$l$ and~$m$ are very close. In such cases, we have an approximation algorithm, as stated by the following proposition. 

\SetCommentSty{rmfamily}
\setlength{\textfloatsep}{2pt}
\begin{algorithm}[t]
   \SetNlSty{scriptsize}{}{}
   \SetNlSkip{0mm}
   \SetKwProg{myfunc}{ColumnSelection}{
   }{}
   \myfunc{$(T, k ,l, \mathcal{R})$ \textup{\textrm{// Table, dimensions and association rules}}}{
   \nl $\tsub[T^*]\leftarrow \emptyset,~ \name{cov}^*\leftarrow -1$\;
   \lnl{ln:enum} \For{$U'\subseteq U$ such that $\card{U}=l$ }{
        \nl $T'\leftarrow \pi_{_{U'}}T$\tcp*{Projection of $T$ on $U'$}
        \nl $\tsub, \name{cov}\leftarrow\textbf{GreedyRowSelection}(T', k, \mathcal{R})$\;
        \nl \lIf{$\name{cov}> \name{cov}^*$}{$\name{cov}^*\leftarrow \name{cov},~ \tsub[T^*]\leftarrow\tsub$}
    }
    \nl \Return $\tsub[T^*]$\;}{}
    \vspace{1mm}
   \SetKwProg{myfunc}{GreedyRowSelection}{}{}
   \myfunc{$(T',k,\mathcal{R})$}{
   \nl $\tsub[T^{**}]\leftarrow\emptyset,~ \name{cov}^{**}\leftarrow-1$\;
   \nl \For{$i$ in $1\dots k$}{
   \nl $\tsub[T^*]\leftarrow\tsub[T^{**}],~ \name{cov}^*\leftarrow\name{cov}^{**}$\;
   \nl \For{$t\in T'-\tsub[T^*]$}{
    \nl $\name{cov}\leftarrow\funct{cellCov_{{\mathcal{R}}}}{T,\tsub}$\;
    \nl  \lIf{$\name{cov}> \name{cov}^*$}{$\name{cov}^*\leftarrow \name{cov},~ \tsub[T^*]\leftarrow\tsub$}
   }
   \nl $\name{cov}^{**}\leftarrow \name{cov}^*,~ \tsub[T^{**}]\leftarrow\tsub[T^*]$\;
    }
   \nl\Return $\tsub[T^{**}],\name{cov}^{**}$\;}
\caption{Greedy Sub-Table Selection.}
 \label{alg:greedy}
\end{algorithm}

\begin{proposition}\label{prop:greedyapprox}
Given a table $T$ over $U$, where $\card{T}=n$ and $\card{U}=m$ and a set $\mathcal{R}$ of association rules, Algorithm~\ref{alg:greedy} computes a $t\times l$ sub-table $\tsub$ such that $\funct{cellCov_{{\mathcal{R}}}}{T,\tsub}\geq(1-\frac{1}{e})\name{OPT}$ where $\name{OPT}$ is the score of the optimal solution to  \texttt{Max-Cell-Cover}.
\end{proposition}
\begin{proof}
Algorithm~\ref{alg:greedy} includes two functions. The ColumnSelection enumerates over the possible column selections and for each computes a sub-table using the function GreedyRowSelection. The latter function iteratively attempts to add each single row to the current sub-table, computes the cell coverage score and records the sub-table with maximal cell coverage. This is repeated $k$ times to select~$k$ rows in total. Thus, for each column selection we greedily compute a sub-table, and among all column combinations we take the best one.

To prove the approximation bound, note that for a fixed set of attributes, the $\name{cellCov_{{\mathcal{R}}}}$ function is non-negative, monotone and submodular with respect to tuples (adding a tuple only increases the score, and as we add tuples the marginal contribution of tuples can only decrease). Therefore, by the well-known result of~\cite{nemhauser1978analysis} a greedy algorithm approximates the optimum by the above-mentioned multiplicative factor. Since we enumerate over all column combinations of size~$l$, we achieve this ratio in particular for the same column selection as the optimal sub-table with cell coverage $\name{OPT}$.
\end{proof}

By the analysis above, the greedy Algorithm~\ref{alg:greedy} is clearly not feasible in the general case due to the need to enumerate all ${m \choose l}$ options for column selection. Note that we cannot greedily select columns, since  the cell coverage metric is \emph{not} sub-modular with respect to columns, due to multi-column association rules.
The algorithm also does not take diversity into account.

In Section~6 we show that even a ``semi-greedy'' variation of Algorithm~\ref{alg:greedy}, which traverses the column combinations in a random order, can take more than two days to run on an industrial-grade server and is therefore still impractical.  Furthermore, halting the algorithm after some fixed time period before enumerating all ${m \choose l}$ possibilities for column selection is not guaranteed to meet the approximation guarantee.

%% file: sections/algorithm.tex
\section{Practical solution}
\label{sec:sol}
As discussed in Section~\ref{sec: metrics}, an ideal sub-table captures a large {and} diverse set of patterns in the full table.
However, optimizing based on a combined score of \textit{cell coverage} and \textit{diversity} is NP-hard, and even approximated solutions are impractical for interactive exploration (Section~\ref{sec:complexity}). 
Therefore, to generate good sub-tables in interactive times (up to several seconds) we take a different approach, based on {\em tabular embeddings}.  The embedding captures bin co-occurrences, and therefore roughly corresponds to frequent itemsets and association rules.

We begin by describing our sub-table selection algorithm, and then discuss its benefits compared to other approaches. 

\subsection{Sub-table Selection}
\label{sec: alg}

\setlength{\textfloatsep}{2pt}
\begin{algorithm}[t]
   \SetNlSty{scriptsize}{}{}
   \SetNlSkip{0mm}
  \SetKwProg{myfunc}{Pre-processing}{}{}
      \myfunc{$(\tilde{T})$  \textup{\textrm{// Raw table}}}{
      \nl $T\leftarrow$ normalize and bin $\tilde{T}$\;

          \vspace{1mm}

   \nl $S\leftarrow$ rows and columns of $T$ as text\;
   \nl $\mathcal{M}\leftarrow\textcd{Word2Vec}(S,\textcd{windowSize}=\max\{n,m\})$\ \tcp*{Embedding Computation}
   \lnl{ln:embed} \Return $\mathcal{M}$ \tcp*{cell-to-vector model: $\mathcal{M}:T\times U\To \mathbb{R}^\gamma$}}
    \vspace{1mm}
   \SetKwProg{myfunc}{Centroid-based Selection}{}{}
   \myfunc{$(T,k,l,Q,U^*,\mathcal{M})$}{
   \nl $\name{rowVecs},\name{colVecs}\leftarrow$ empty dictionaries\;
   \lnl{ln:query} \lIf{$Q\neq\textcd{NULL}$}{$T\leftarrow Q(T)$}
   \nl $U\leftarrow$ columns of $T$\;
   \lnl{ln:avgrowsbegin}\For{$t\in T$} {
   \nl      $v\leftarrow \funct{avg_{u\in U}}{\mathcal{M}(t(u))}$\;
   \lnl{ln:avgrowsend}      $\name{rowVecs}\leftarrow\name{rowVecs}\cup \{v\mapsto t\}$\;
   }
   \nl $\mathcal{C}\leftarrow\textcd{cluster}(\name{rowVecs},k)$\;
   \nl $\tsub\leftarrow\name{rowVecs.getValues}(\textcd{centroids}(\mathcal{C}))$\;
   \lnl{ln:avgcolsbegin}\For{$u\in U-U^*$} {
   \nl      $v\leftarrow \funct{avg_{t\in T}}{\mathcal{M}(t(u))}$\;
   \lnl{ln:avgcolsend}      $\name{colVecs}\leftarrow\name{colVecs}\cup \{v\mapsto u\}$\;
   }
   \nl $\mathcal{C}\leftarrow\textcd{cluster}(\name{colVecs},l-\card{U^*})$\;
   \nl $\tsub[U]\leftarrow U^*\cup\name{colVecs.getValues}(\textcd{centroids}(\mathcal{C}))$\;
   \nl $\tsub \leftarrow \Pi_{_{\tsub[U]}}\tsub$\;
   \nl\Return $\tsub,\tsub[U]$\;}

\caption{\subtab{} Algorithm for Sub-Table Selection.}
 \label{alg:tsub}
\end{algorithm}

Our algorithm, shown in Algorithm~\ref{alg:tsub}, includes two parts: (1) Pre-processing, in which we compute a vector representation for each cell in the full table $T$ using \textit{table embedding}, and (2) centroid-based sub-table selection, which utilizes the embedded vectors to quickly select a sub-table.
Importantly, pre-processing  is performed only once, when the table $T$ is loaded,
and centroid selection is performed  for each exploratory query that the analyst performs over the table $T$ to produce the corresponding sub-table.

\paragraph*{Pre-Processing}
Given a raw table $T$, the first step  

is to normalize the values (e.g., remove illegal characters) and bin continuous columns so that values are replaced by their bin name (e.g., binning splits the  \tcol{Distance} column into short, medium and long distances). Let $\tilde{T}$ be the normalized, binned table. 

We then use a {\em table embedding process} to generate a real-valued vector for each cell in the table $\tilde{T}$. Note that while several recent works~\cite{embeddings_integration_tasks, deng2020turl, tang2020rpt, Table2Vec2019,Cappuzzo_2020}
suggest more complex methods for embedding tabular data, e.g., based on graph representations, auto-encoders, etc., we use a simpler yet effective method that quickly computes the vector representations based on~\cite{bordawekar2019exploiting}. (See  Section~\ref{sec:quality_analysis} 
for a comparison with~\cite{Cappuzzo_2020}).

Our embedding method transforms the table into a corpus of \textit{sentences} and then uses a fast implementation of word embedding~\cite{mikolov2013distributed}.
The corpus consists of \textit{tabular sentences} in which each cell in the table $T$ represents a single word. We use two types of sentences: \textit{tuple-sentences}, containing values in each tuple $t\in T$, and \textit{column-sentences} that cover the values in $T(u),\forall u \in U$.
To achieve even faster execution times, we limit the corpus size to $100K$, where the sentences are chosen uniformly at random. 

Using the corpus of tabular sentences we then train a word-embedding~\cite{mikolov2013distributed} model that outputs a numeric vector representation for the table cells. The learned representation of the cells are based on co-occurrences of values in the same row/column, and hence capture recurring patterns which roughly correspond to association rules. 
The output of this process, as depicted in Line~\ref{ln:embed}, is a mapping between each cell in $T_{ij}$ to a corresponding, learned \textit{cell-vector} $\mathcal{M}_{ij}$.

\paragraph*{Centroid-Based Sub-table Selection}
Once the vector representations $\mathcal{M}$ are computed for each cell in $T$, we perform a fast yet effective sub-table selection based on the vectors' centroids. As mentioned above, this is done over the results of each exploratory query performed by the user.

We select a sub-table $T_{sub}$ of size $k\times l$ as follows:
First, to select the $k$ tuples in $T_{sub}$ we compute for each tuple $t\in T$ with columns $U$, a \textit{tuple-vector}, by taking the component-wise average of its corresponding cell-vectors.
Namely, we average the cell-vectors $\mathcal{M}\left(t(u_1)\right),\mathcal{M}\left(t(u_2)\right),\dots,\mathcal{M}(u_{|U|}))$

(lines~\ref{ln:avgrowsbegin}-\ref{ln:avgrowsend}).

This allows us to obtain a unified representation of each tuple in the input table. 
We then cluster the tuple-vectors into~$k$ clusters and select their centroids as the rows of $\tsub$. 
Next, to select the columns of $\tsub[U]$ we perform a similar process, creating \textit{column-vectors}, forming clusters, and finding their centroids. 

Since the columns of $U^*$ must be included in the sub-table, we exclude them from clustering, compute only $l-\card{U^*}$ clusters and then add the $U^*$  columns to the selected centroids.

\subsection{Discussion}
We conclude this section with three important observations regarding our embedding-based approach for sub-table generation, which explain why our system works well in practice despite the fact that it is not explicitly based on association rules.

First, note that since \subtab{} does not directly attempt to optimize the cell-coverage and diversity metrics,
there are no guarantees it will always obtain
high scores. 

However, as our experiments will show (Section~\ref{sec: exp}), \subtab{} computes high quality sub-tables in terms of both cell-coverage and diversity when the underlying rules in $T$ are prominent (rather than arbitrary)

We also show that directly optimizing the metrics, using either a more-feasible \textit{semi-greedy} variation of Algorithm~\ref{alg:greedy} or a \textit{Multi Armed Bandit} sampling algorithm takes over \textit{24 hours} to achieve the scores \subtab{} obtains in several seconds.

Intuitively, the reason that our embedding-based algorithm achieves good results is twofold. (1) \subtab{} achieve good cell-coverage scores since the embedded vectors, which are generated based on the frequency of co-occurrences of values in the same rows and columns, capture underlying frequent patterns -- as is also done in association rules mining. (2) We achieve a high score for diversity due to the centroid-based columns and tuples selection. 

As mentioned earlier, obtaining vector representations for table cells could be done using several different techniques~\cite{bordawekar2019exploiting,embeddings_integration_tasks,deng2020turl, tang2020rpt, Table2Vec2019,Cappuzzo_2020}.
In our experimental evaluation, we compare \subtab{} to a baseline which uses the embedding method suggested in ~\cite{Cappuzzo_2020}, which is geared towards data integration tasks.
We show that not only does \subtab{} complete the pre-processing phase 26X faster (90 seconds, rather than 40 minutes), but that it also obtains superior scores in terms of cell-coverage and diversity.

Finally, note that it is also possible to apply clustering directly on $T$, without first generating an embedded representation. This method is also inferior to \subtab{} as 
it relies on a ``one-hot-encoding'' of the data, which does not capture the underlying patterns as well as the embedding-based method (see Section~\ref{sec: exp}).

%% file: sections/experiments.tex
\section{Experiments}
\label{sec: exp}

We performed an extensive experimental evaluation of \subtab{} in terms of both the quality and usefulness of the resulted sub-tables as well as its running-times.
After describing the experimental setup, we report our results. First, to test the sub-table quality (Section~\ref{sec:quality_analysis}), we conducted two sets of experiments: (1) a twofold user study, where the participants not only rank the usefulness of the sub-tables through a questionnaire, but are also required to list insights and conclusions they identify in each sub-table; and
(2) an offline, simulation-based study, in which we retraced real-life analysis sessions, generated a sub-table for each exploratory query, then checked whether the parameters of the next query in the session (e.g., selection term, aggregation column) appear in the sub-table. We also examined whether our combined metric of cell-coverage and diversity is correlated with these quality evaluations.

The final two experiments measure the running-time of \subtab{} for each dataset (Section~\ref{sec:running_times}), and the performance using different parameter settings for the packages used in \subtab{} (Section~\ref{sec: parameters_tuning}). 
A summary of findings from the experiments are in Section~\ref{sec: findings}.

\subsection{Experimental Setup}\label{sec:setup}

\subtab{} is implemented in Python~3.8
as a local Python library that hooks into Pandas~\cite{pandas} and therefore can be used, e.g.\  in common EDA environments such as Jupyter notebooks.
The binning method used is based on kernel density estimation and is implemented with \emph{sciPy}\footnote{https://scipy.org/}. The Word2Vec embedding method is implemented by \emph{gensim}\footnote{https://radimrehurek.com/gensim/}. Centroid selection is performed by creating clusters via KMeans using \emph{sklearn}\footnote{https://scikit-learn.org/}.The experiments were run on Intel Xeon CPU based server with 24 cores and 96 GB of RAM.

\paragraph*{Metrics implementation}
As the cell coverage metric relies on association rules, we provide details about how they are mined. We compute the association rules using the Apriori Algorithm \cite{agrawal1994fast} and implement it using \emph{efficient-apriori}\footnote{https://pypi.org/project/efficient-apriori/}; we set the main parameters, support and confidence, to $0.1$ and $0.6$, respectively, and the minimum rule size to 3. In Section \ref{sec: parameters_tuning}, we conduct several experiments that test the effect of each of the parameters, by varying the given parameter while setting all others to their default values. 
When target columns are selected by the user, the data is split according to the binned values of the target columns.  The rules are then mined over each subset separately. An optional extension to our system is coloring the patterns (association rules) in the data that are represented by the sub-table, as illustrated in Figure~\ref{fig:miniex}; this was found very helpful by participants in the user-study (Section \ref{sec:quality_analysis}). 
Lastly, to implement our combined cell coverage and diversity metric, we take $\alpha = 0.5$ by default for the combined score, assigning equal weights to cell coverage and diversity.

\begin{table}[]
    \centering
    \begin{tabular}{c|c|c|c}
        Metric & SubTab & RAN & NC \\
        \hline

         \# correct insights & \textbf{4} (85\%) & 1.2 (30\%)& 0.2 (6\%)\\
         \% of users with no insights & \textbf{0}\% & 12\% & 89\%\\
         \# Total insights &	\textbf{4.5}&	3.67&	1.5\\

    \end{tabular}
    \caption{Results of the user study }
    \label{tab:Results of the user study}
\end{table}

\begin{figure}
 \includegraphics[width=0.8\columnwidth] {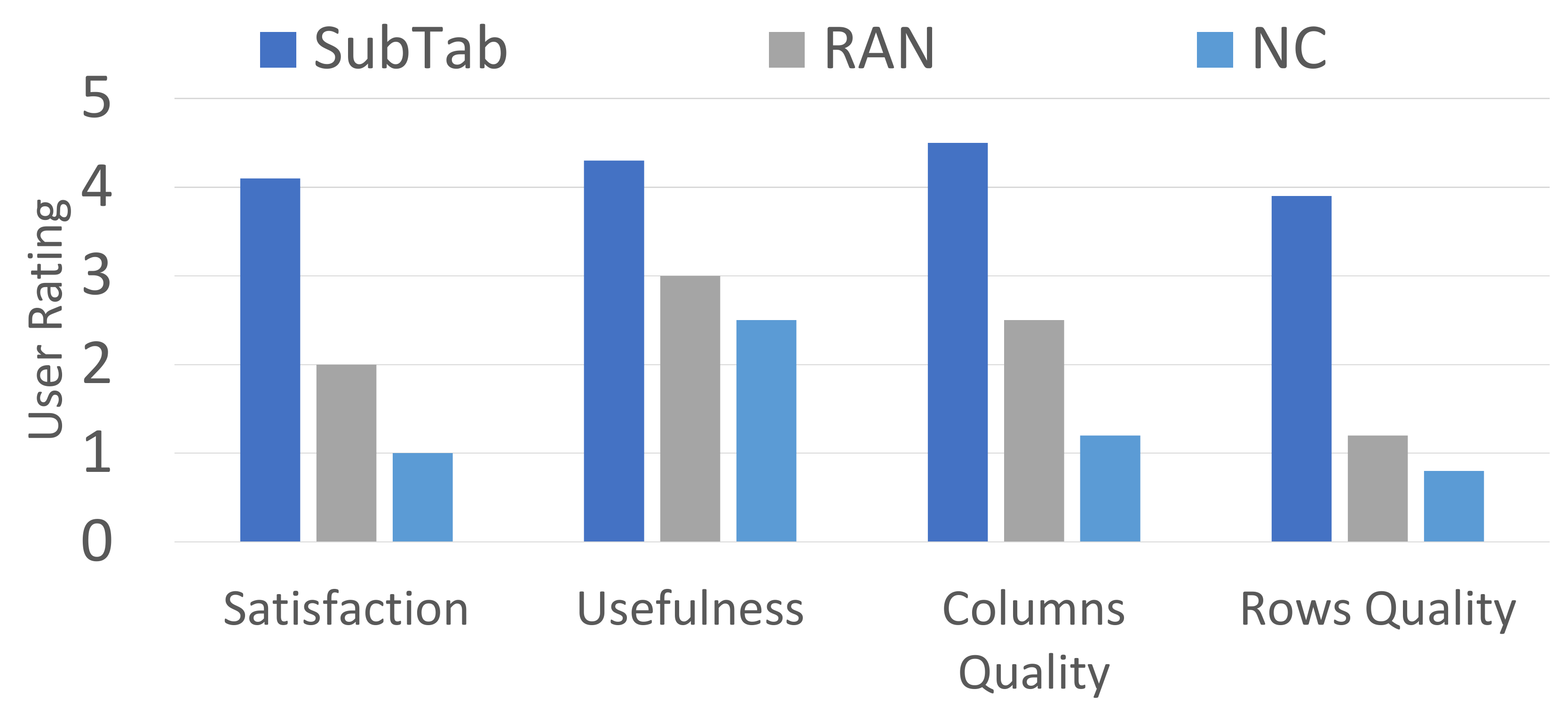}
\caption{Users overall ratings}
\label{fig:user_ranking} 
\end{figure}

\begin{figure}
    \centering
    \includegraphics[width=0.8\linewidth] {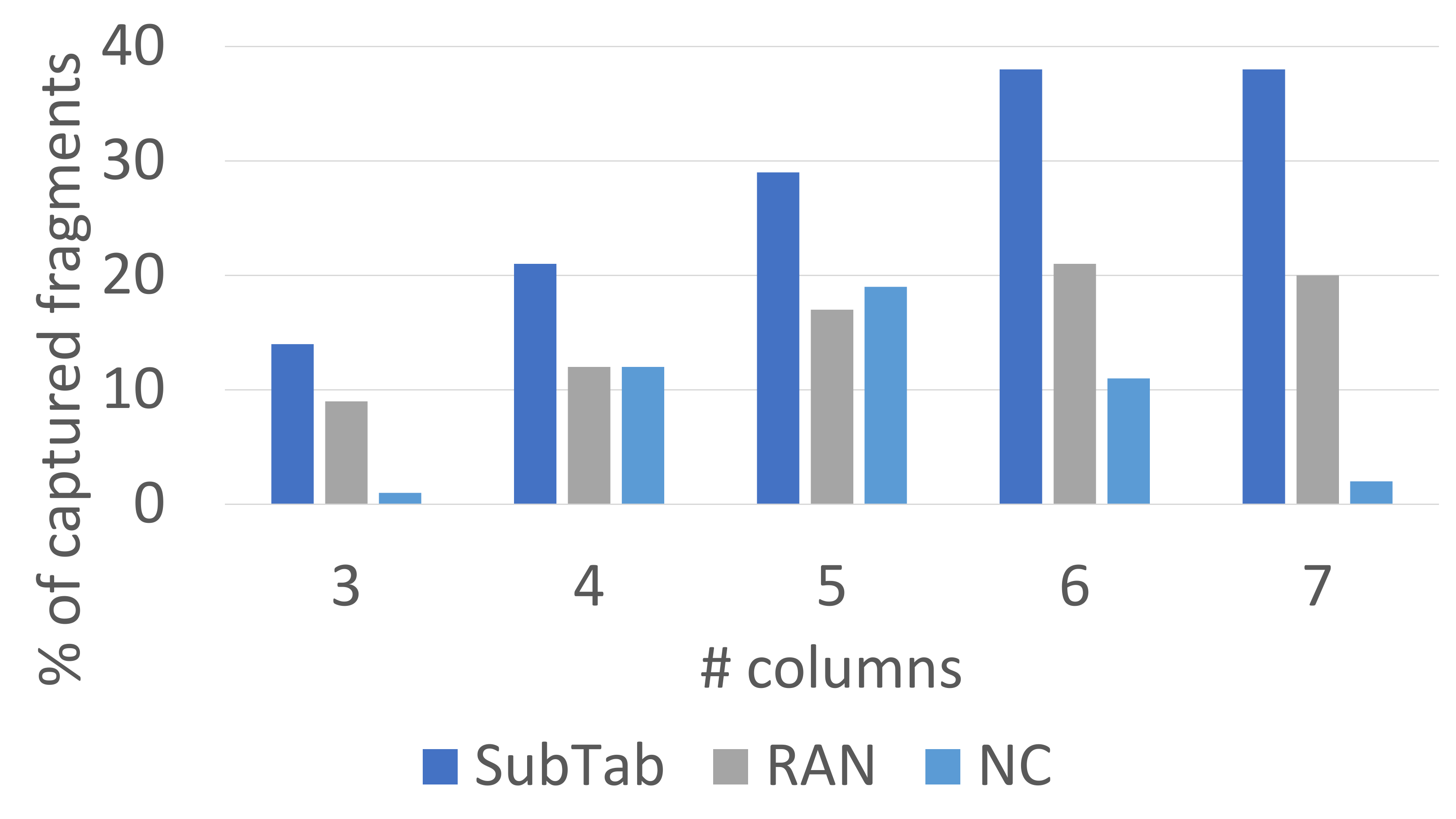}
    \caption{Simulated Experiments Results ($CY$)}
    \label{fig:user_study_cyber_cor}    
\end{figure}

\paragraph{{Datasets.}}
To demonstrate the performance of \subtab{} in different domains, we used the following datasets: 
\begin{itemize}
    \item Flights ($FL$), \footnote{https://www.kaggle.com/usdot/flight-delays?select=flights.csv} which has 6M rows and 32 columns
    \item Cyber-security 
    ($CY$)\footnote{https://www.honeynet.org/challenges/}, which has 30K rows and 15 columns 
    \item Spotify ($SP$), \footnote{https://www.kaggle.com/c/bfh-spotify-challenge/data} which has 42K rows and 15 columns, and 
    \item Credit card  frauds ($CC$), \footnote{https://www.kaggle.com/mlg-ulb/creditcardfraud} which has 250K rows and 31 columns.
    \item US Funds ($USF$), \footnote{https://www.kaggle.com/stefanoleone992/mutual-funds-and-etfs?select=MutualFunds.csv} which has 23.5K rows and 298 columns. 
    \item Bank Loans dataset (BL) \footnote{https://www.kaggle.com/panamby/bank-loan-status-dataset} with 110K rows 19 columns.
\end{itemize}

\paragraph{{Baselines.}}
We have tested two types of baselines. 
The following baselines support fast response time and are suitable for EDA. They were used in our quality analysis and user study.
\begin{enumerate}
    \item Random ($RAN$): select uniformly at random~$k$ rows and~$l$ columns. To increase the quality of this baseline, we iteratively repeat the random selection for one minute, and return the sub-table with highest score among all the randomly drawn sub-tables. 
    \item Naive clustering ($NC$): 

   We first transform the categorical and textual columns to be continuous values using one-hot encoding\footnote{https://scikit-learn.org/stable/modules/generated/sklearn.preprocessing.OneHotEncoder.html}. Then we treat each row as a vector of length $m$ (the total number of columns), and cluster the vectors using K-means.  The centroids of those clusters are used as the rows in the sub-table. We select the columns analogously.

\end{enumerate}

The second set of baselines that we test are too slow to be used in an interactive setting, having  running times of over~30 minutes.
We nevertheless use them for comparing the \emph{quality} of our system. 
\begin{enumerate} \addtocounter{enumi}{3}

    \item Multi-Armed Bandit ($MAB$): we use a version of the Multi-Armed Bandit algorithm \cite{intro_MAB} that, in each iteration, selects a set of $n$ rows and $k$ columns and evaluates the sub-table using our metric.  The reward (i.e.\ the cell coverage score) is given to all the columns and rows that participated in the sub-table, and the exploration-exploitation method used is Upper Confidence Bound (UCB)~\cite{UCB_MAB}. 
    \item Greedy sub-table selection (Greedy): We modify the greedy algorithm outlined in Algorithm \ref{alg:greedy} by traversing the column combinations in random order
    (line~\ref{ln:enum}). We can thus halt the algorithm after any number of iterations and use the sub-table with maximal score among the ones found until that point. For our experiments, we use a time limit of 5 hours, which we empirically found to be required for discovering informative sub-tables.

    \item $EmbDI$:\cite{Cappuzzo_2020}: This algorithm creates a local embedding (inspired by Node2Vec) that is effective for data integration tasks in relational databases. The table is transformed into a graph by representing the columns and rows as nodes, connected by edges which capture structural relationships.  The structure created provides a more efficient graph computation than the naive graph transformation. 
    Thus, it captures relationships inherent in the tables, and provides a more efficient computation than the naive graph transformation. 
\end{enumerate}

\subsection{Quality Analysis} \label{sec:quality_analysis}

To assess the quality of \subtab{}, we  evaluated the sub-tables generated by \subtab{}, using several experiments: (1) We performed a ``live'' user-study, where participants used \subtab{} as well as other baselines in real-life analysis tasks;  (2) Offline simulation-based evaluation, where we retraced completed EDA sessions and assessed the usability of a sub-table by determining if it contains elements of the subsequent query; and (3) Evaluation based on our metrics -  to further compare \subtab{} to the baselines, we evaluated the generated sub-tables in terms of our cell coverage and diversity metrics to test whether our metric score is highly correlated with the user-study and simulation-based evaluation. In this part, we also compare our results to slower baselines that could not be tested in the live user study.
These experiments are discussed below.

\subsubsection{User Study}\label{sec:userstudy}
We conducted a user study to compare the usefulness of \subtab{} to the baseline approaches.
We recruited 15 participants, all with varying degree of expertise in data analysis using Pandas. The participants were first asked to perform an actual data-analysis task in which they used the sub-table to discover insights about a dataset, and then answer a short questionnaire about the quality of the sub-tables. We divided the participants to equal groups, each worked with a different baseline on three different datasets: $SP$, $FL$ and $BL$ . As this is a ``live'' experiment, we compared \subtab{} with the baselines $RAN$ and $NC$, which are fast enough for interactive analysis. 
 
For each dataset there was an exploration task involving several queries. Each user was given one baseline, and performed explorations over all three datasets. The user's goal was to write down insights that are relevant to the given task while examining the sub-tables that were created during the exploration. We then counted the number of correct insights that users had, and averaged the results for each exploration task, per user. 
To help the users in their exploration task, we also colored the patterns (association rules) that were captured in the sub-table for all the baselines, using the $SP$ and $FL$ dataset. In the $BL$ exploration, we did not color the patterns (in any baseline), displaying only the created sub-table without additional information. By doing this, we tested whether the trends of the exploration task without coloring remains the same as those with coloring.

\paragraph{Insights Discovery Experiment} 
For each dataset, the participants were presented a notebook containing several exploratory queries, and a \textit{sub-table} of the queries' results (generated by either \subtab{}, $RAN$, or $NC$). The participants were then instructed to examine each notebook and derive insights about the dataset's particular analysis task. For example, the task in $SP$, containing data about songs and their popularity in the Spotify streaming service, was to discover \textit{``what makes songs popular''}. 
We then manually evaluated the \textit{correctness} of the participants' insights, and removed ones that were statistically incorrect or highly irrelevant to the analysis task.

Table \ref{tab:Results of the user study} shows the total number of insights, percentage of correct insights, as well as the percentage of users who did not derive insights at all, averaged across all three datasets.
First, see that when using \subtab{}, users derived an average of $4$ correct insights per dataset, which is 3X more than $RAN$ and 12X more than $NC$.
Interestingly, the percentages of correct insights obtained by \subtab{} is 85\% whereas only 30\% and 6\% of the insights obtained via $RAN$ and $NC$ (resp.) were correct.
Inspecting the incorrect insights, we observed that the users reached false conclusions since many of the sub-tables produced by $RAN$ and $NC$ were \textit{misleading}. These sub-table contained, for example, non-representative distribution of columns, or presented a random, false correlation between columns. 
Last, see that using \subtab{} 100\% of users were able to successfully finish the analysis task, whereas none of the users failed to derive at least one insight.

\begin{figure}
    \includegraphics[width=0.9\columnwidth] {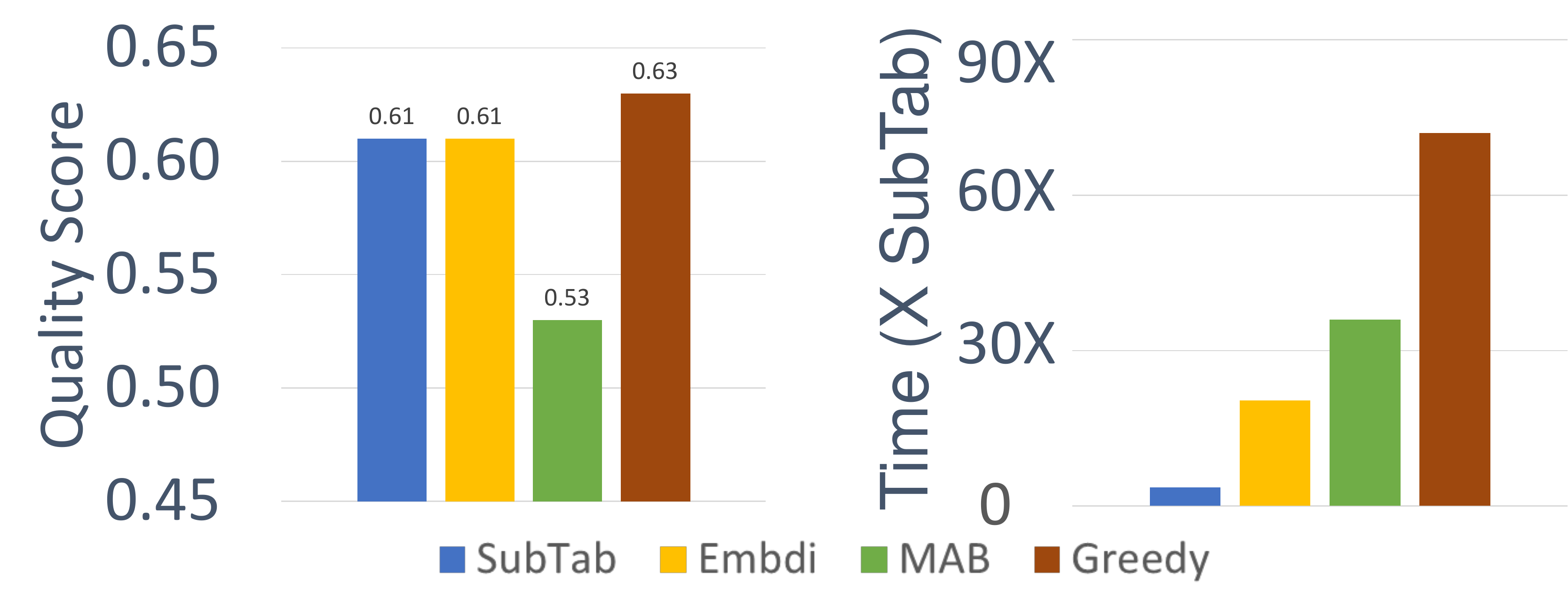}
    \caption{(a) Quality score and (b) Total running time}
    \label{fig:running_time_And_quality}    
\end{figure}  

\begin{figure*}
     \centering
     \begin{subfigure}[b]{0.32\linewidth}
         \centering
         \includegraphics[width=\linewidth]{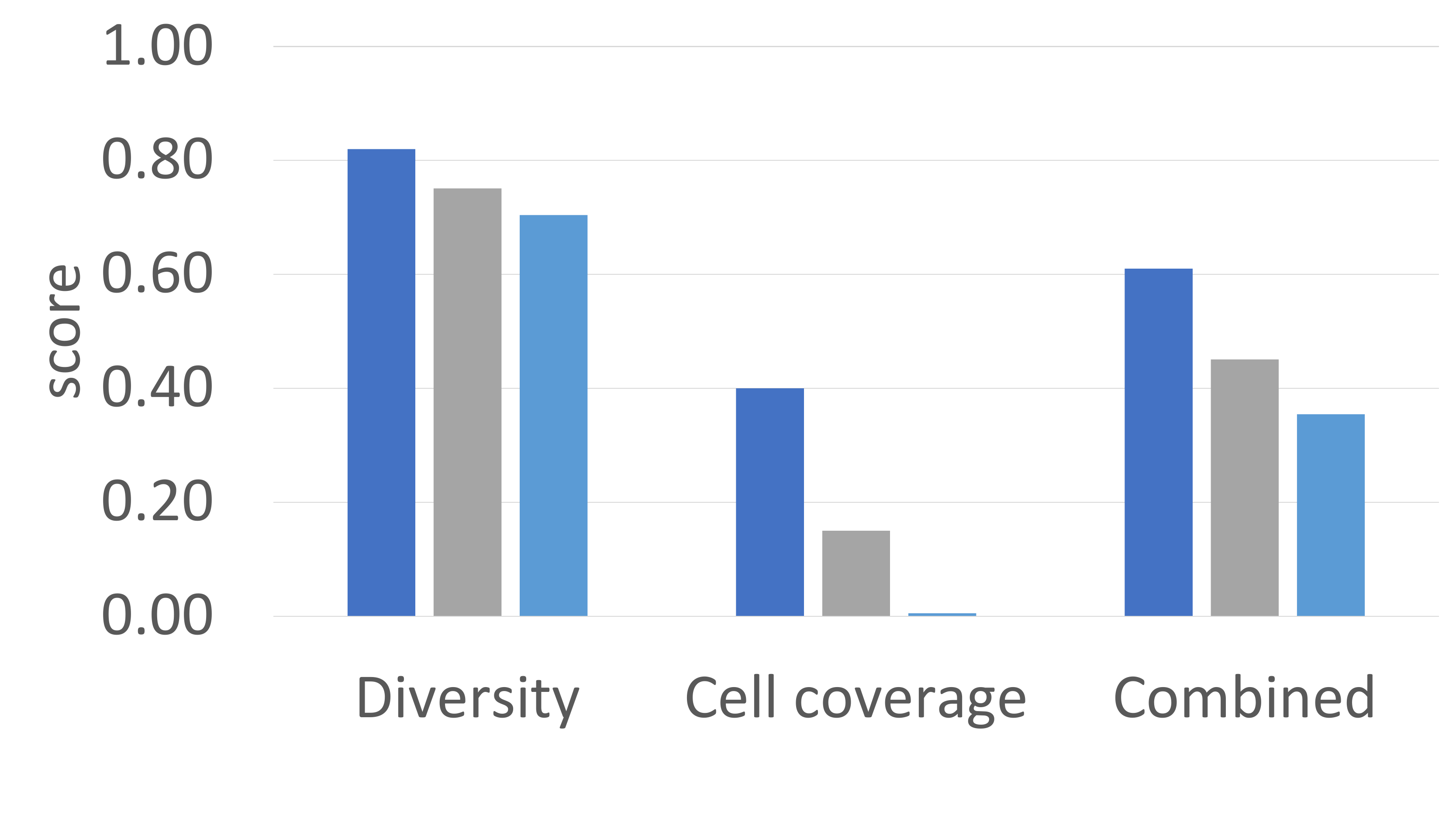}
         \caption{FL}
         \label{fig:quality_flights}
     \end{subfigure}
     \begin{subfigure}[b]{0.32\linewidth}
         \centering
         \includegraphics[width=\linewidth]{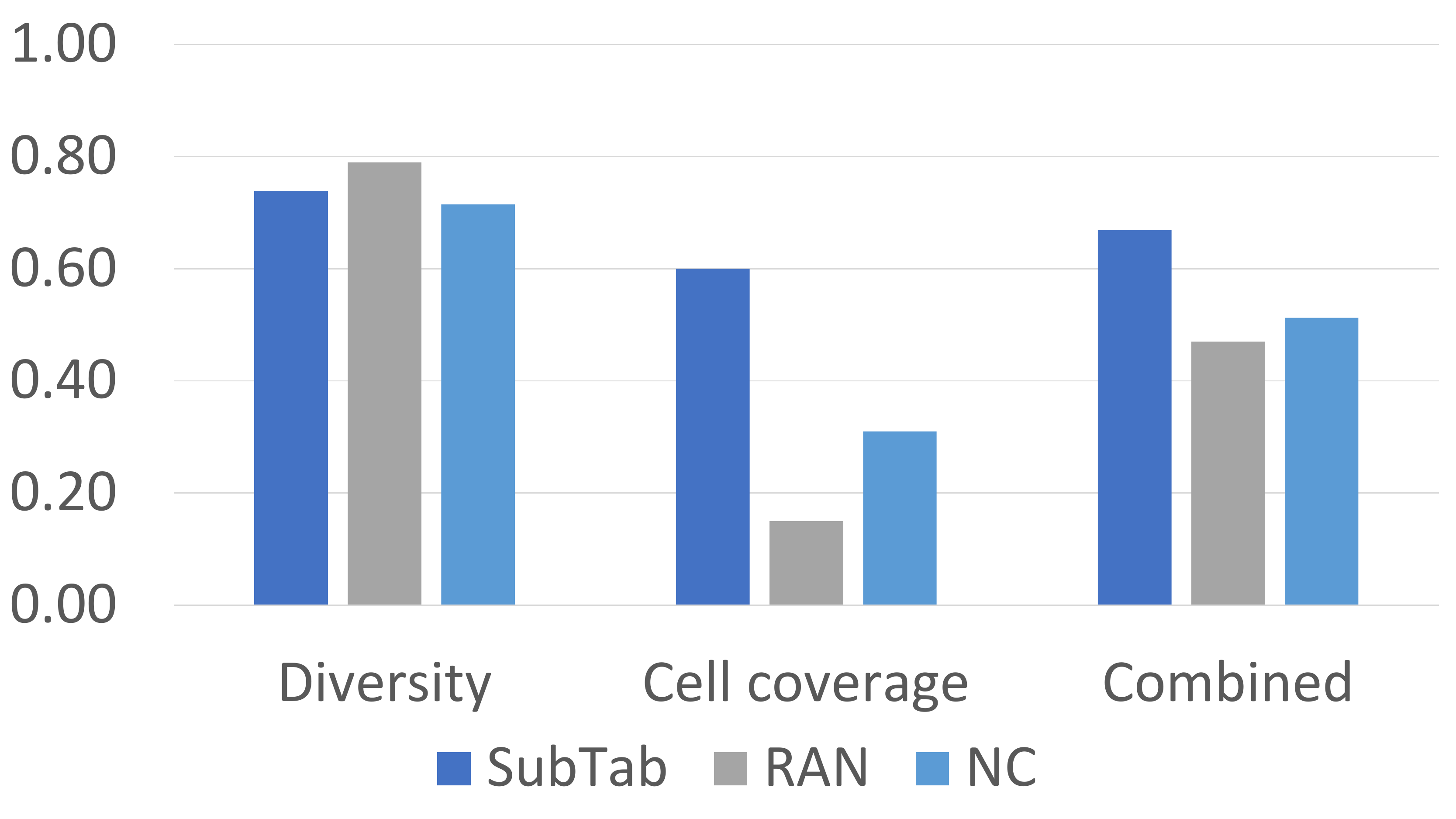}
         \caption{SP}
         \label{fig:quality_music}
     \end{subfigure}
     \hfill
     \begin{subfigure}[b]{0.32\linewidth}
         \centering
         \includegraphics[width=\linewidth]{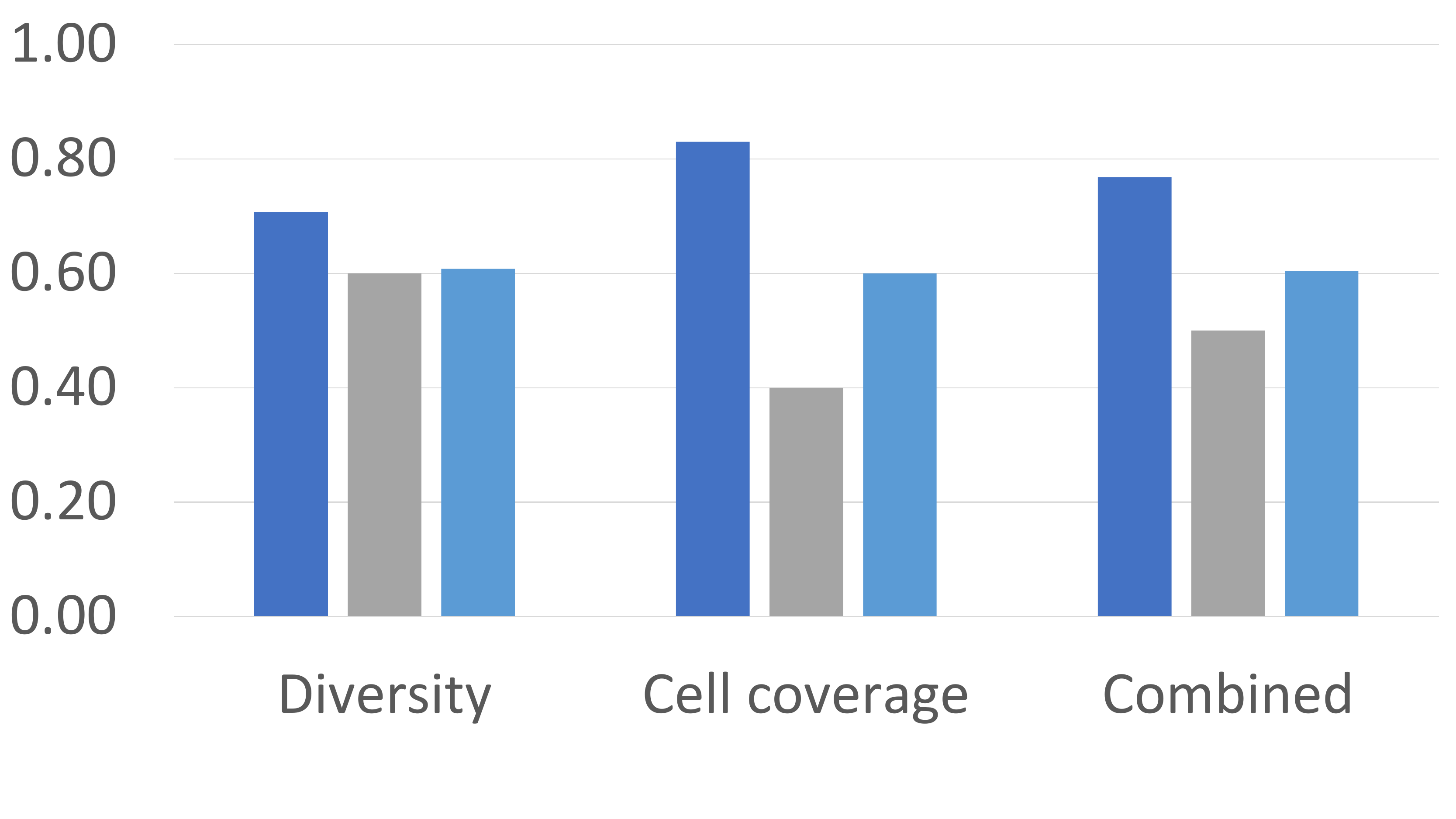}
         \caption{CY}
         \label{fig:Quality_Cyber}
     \end{subfigure}
        \caption{quality metrics for different baselines and datasets }
        \label{fig:quality_all}
\end{figure*}

\paragraph{Questionnaire Results.} 
As mentioned above, after the participants completed the insights discovery tasks, they were each given a short questionnaire asked to evaluate the sub-tables on a scale of 1 (strongly disagree) to 5 (strongly agree) according to the following statements:

\begin{itemize}
    \item Q1: The presented system is better than the standard dataframe sub-table.
    \item Q2: Would you like to use the sub-table system in future data exploration tasks?
    \item Q3: The sub-tables' columns were relevant to the queries.
    \item Q4: The sub tables' rows are representative and capture patterns.

\end{itemize}

The questionnaire results are summarized in Figure~\ref{fig:user_ranking}.
Note that the average users' ranking of \subtab{} is above 4 for all statements, and are significantly higher than $RAN$ and $NC$.

\subsubsection{Simulation-Based Study.}\label{sec:simulation}
We conducted an additional offline experiment, in order to further evaluate the quality sub-tables. We used a publicly-available collection of 122 data exploration sessions~\cite{react_EDA}, containing select, project, group-by, and sort operations over the dataset $CY$.
To evaluate the potential usefulness of the sub-tables we replayed each query in a session, and generated a corresponding sub-table using \subtab{} and the baselines $RAN$ and $NC$.
We then examined whether the next query in each session contain a fragment (e.g., a group-by attribute, selection term, etc.) that appears in the sub-table of the previous query's results. 
Intuitively, appearance of next-query fragments in the sub-table, may imply that the sub-table is useful in selecting the next exploration step.

The percentage of captured query fragments are shown in Figure \ref{fig:user_study_cyber_cor}, when varying the width (i.e., number of columns) of the sub-table from 3 to 7 (out of the 12 columns of the $CY$ dataset). See again that \subtab{} significantly outperforms the baselines, and is able to predict 14\% (width=3) to 38\% (width=7).
Naturally, the results improve as the sub-table covers more columns, however it is still very difficult to cover all query fragments (e.g., practically any value from a column's value domain can be used as a selection term).

\subsubsection{Quality by Our Metrics} 
So far we have evaluated the quality of sub-tables based on their performance in external tasks. We now compare to our intrinsic quality metrics: cell coverage, diversity and combined score.

\paragraph*{Comparison to the interactive baselines} Figure~\ref{fig:quality_all} shows the three scores for the three baselines over the FL, SP and CY datasets. For the three datasets, \subtab{} achieves a significantly higher cell coverage and combined scores, compared with the baselines. Interestingly, in FL and CY it also achieves a higher diversity score, which means it outperforms the baselines for any choice of $\alpha$. In SP, RAN has a slightly better diversity score, but its cell coverage is extremely low. For example, for the $SP$ dataset, \subtab{} achieves a total score of 0.68, were the $RAN$ and $NC$ achieve 0.47 and 0.51 respectively.

To connect between intrinsic and external metrics, we now compare their ranking of baselines. In our user study (Section~\ref{sec:userstudy}), we compute the combined score for each presented sub-table, and average per baseline. The resulting average scores for \subtab{}, RAN and NC were $0.56$, $0.32$ and $0.15$ respectively, which matches the ranking of these baselines in terms of user ratings (Figure~\ref{fig:user_ranking}). Similarly, for the simulation based study (Section~\ref{sec:simulation}), we have computed the combined score for each computed sub-table, and averaged per baseline and per sub-table size. The resulting ranking between baselines per sub-table size was identical to the ranking according to the percentage of matched steps (Figure~\ref{fig:user_study_cyber_cor}).
This indicates that our metrics correlate with human judgements and with the usefulness of sub-tables in EDA sessions.

\paragraph{Comparison to slower baselines}
We have also executed the slower, non-interactive baselines over the $FL$ dataset, and show, in Figure~\ref{fig:running_time_And_quality}, their performance in terms of quality and time and compared with \subtab{}. \subtab{} achieves the same combined score as the EmbDI baseline; however, the latter takes~40 minutes to execute, whereas \subtab{} requires only~1.5 minutes. MAB achieves the worst quality, even though it is executed for a long time. Finally, the Greedy baseline slightly outperforms the other baselines in terms of quality, but is the slowest one -- this score was achieved by executing it for 48 hours on a multi-process architecture. Overall, this shows that \subtab{} computes high-quality tables at interactive speed.

\subsection{Effectiveness of Pre-processing} \label{sec:running_times}

Recall that \subtab{} has two distinct steps:  Pre-processing, which is executed once upon loading a data table, and Selection, which is executed with each \subtab{} display, both for the table itself and for queries over it (see figure~\ref{fig:arch}). Figure~\ref{fig:running_time} shows the execution times for each step over different datasets. Pre-processing takes the longest time, 90~seconds, for the CC dataset,  although it is smaller than FL. The reason is that this data contains only numeric columns that must undergo binning. Still, this is a reasonable time for the set-up phase of an EDA session. Then, the Selection phase takes only a few seconds for all the datasets. We have tested the computation time for various sub-table sizes, and the results were similar (the difference is less than $10\%$). This shows that  our reuse of embeddings is indeed effective in achieving fast response time.

\begin{figure}
 \includegraphics[width=0.7\columnwidth] {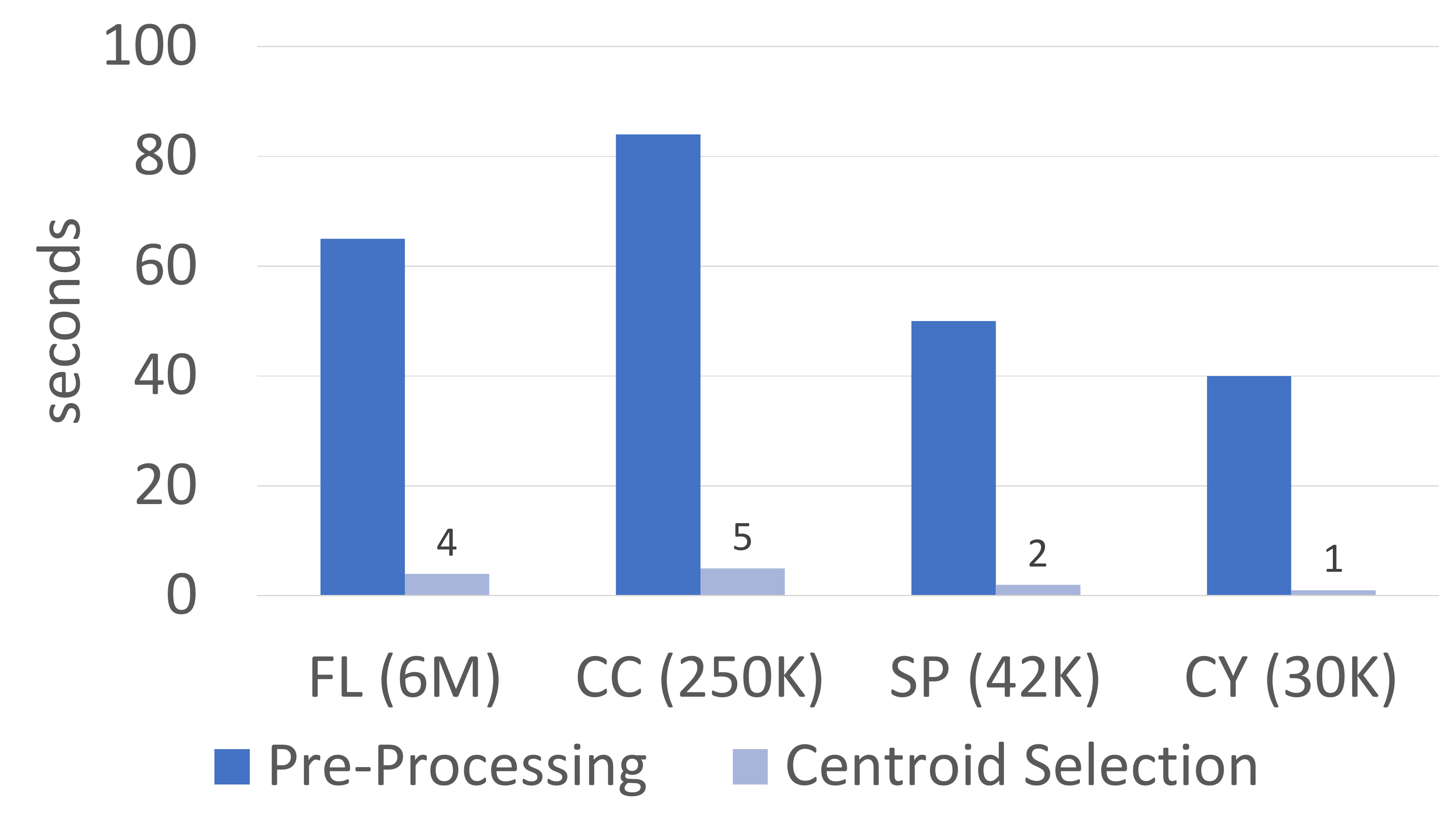}
\caption{Average running time of \subtab{}}
    \label{fig:running_time}    
\end{figure}

\begin{figure*}
\centering
     \begin{subfigure}[b]{0.32\linewidth}
         \centering
         \includegraphics[width=\linewidth]{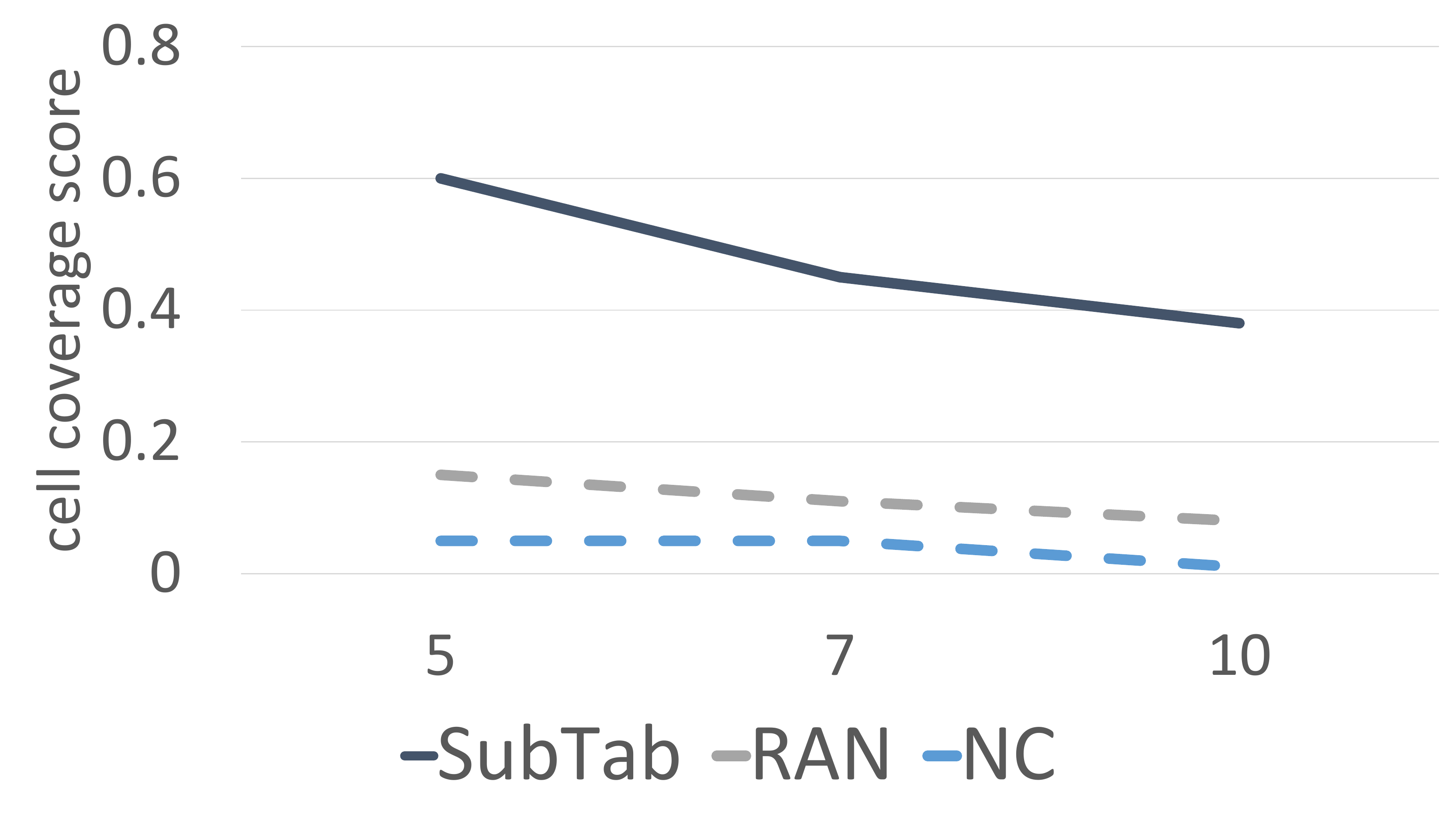}
         \caption{\# Bins}
         \label{fig:binning_methods}
     \end{subfigure}
     \hfill
     \begin{subfigure}[b]{0.32\linewidth}
         \centering
         \includegraphics[width=\linewidth]{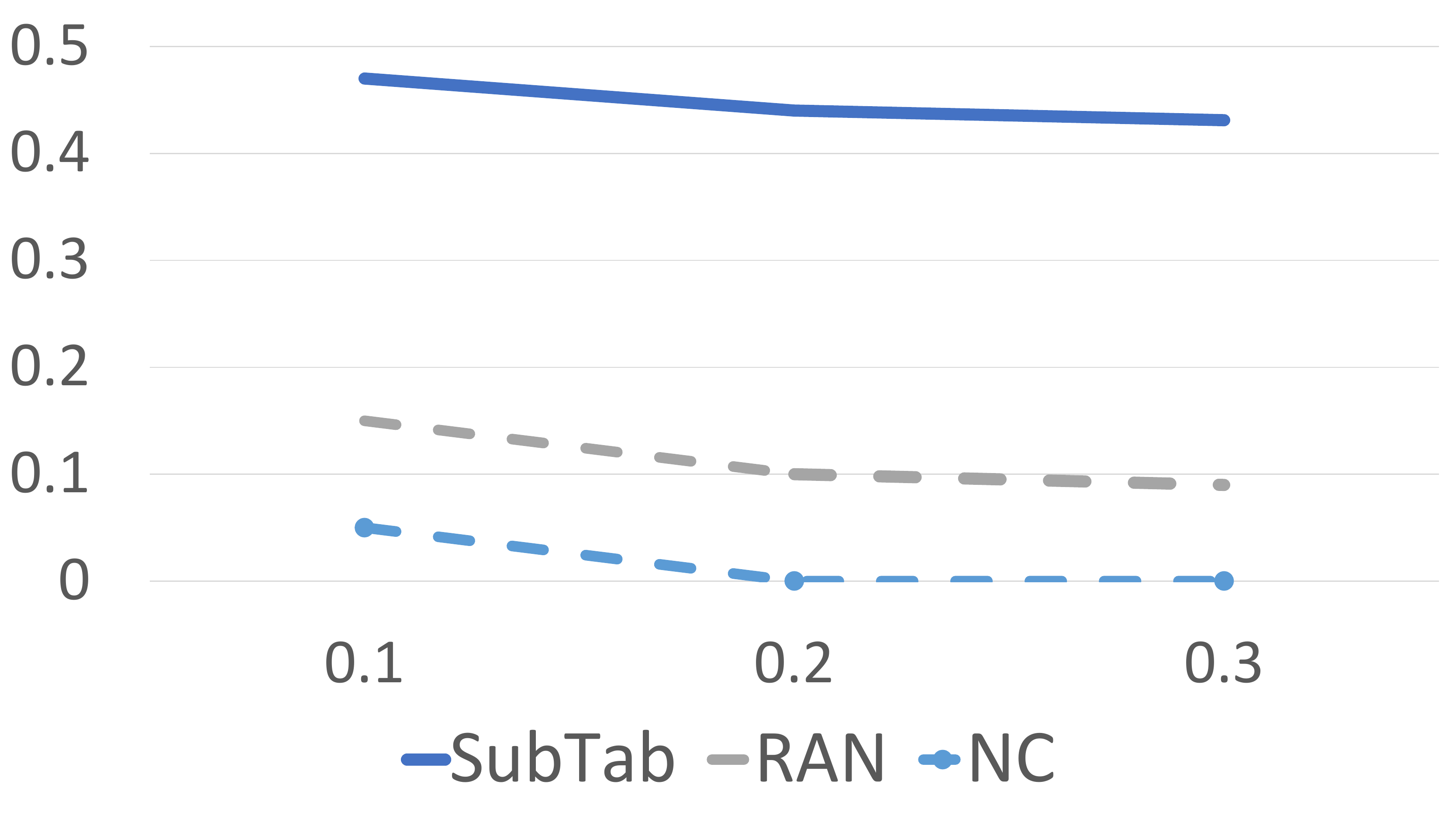}
         \caption{Support Threshold}
         \label{fig:support}
     \end{subfigure}
          \begin{subfigure}[b]{0.32\linewidth}
         \centering
         \includegraphics[width=\linewidth]{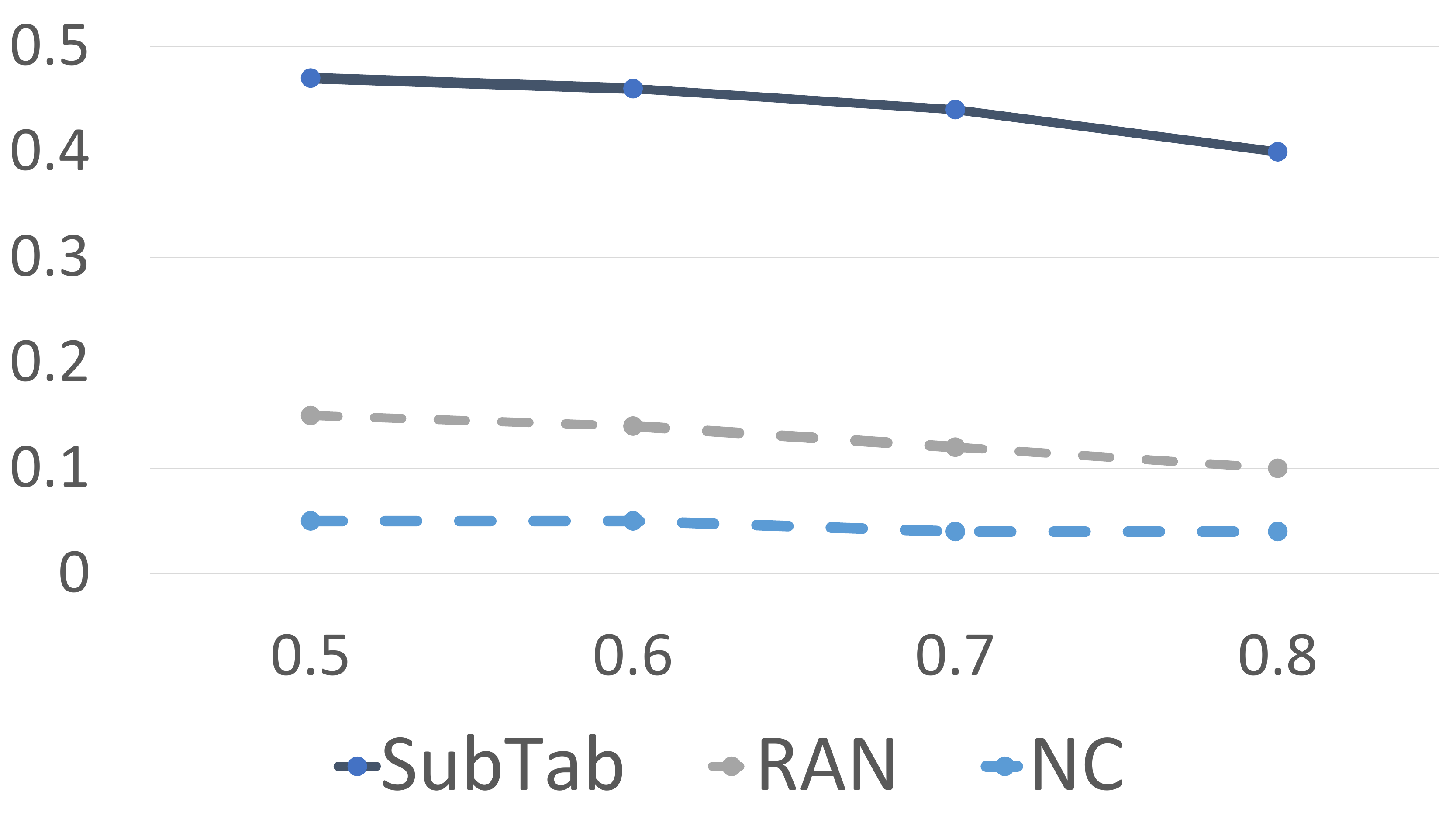}
         \caption{confidence Threshold}
         \label{fig:confidence}
     \end{subfigure}

    \caption{Parameter tuning experiment}
    \label{fig:Parameters_Tuning}    
\end{figure*}

\subsection{Parameter Tuning} \label{sec: parameters_tuning}
Our metric of cell coverage (formula~\ref{formula:coverage}) is measured with respect to an input set of association rules. We have explained above, in Section~\ref{sec: alg}, why the embeddings used by our solution implicitly capture the data patterns defined by prominent association rules. We now unwrap the imprecise notion of ``prominence'' by considering parameters that affect the set of mined association rules, and showing that \subtab{} performs well for varying values for these parameters. The averaged results for $FL$ and $SP$ datasets are shown in Figure~\ref{fig:Parameters_Tuning}. In each graph we vary one parameter, and use the default value for the others. Note that the sub-tables computed by each of the algorithms \emph{are the same} across all settings, since these algorithms do not use the association rules as input; the variation in parameters only affects the means of evaluating the sub-tables.

The first parameter that we consider is the number of bins per binned column (the default number is~5). A larger number of bins would imply more association rules (since there are more bin combinations) but with lower significance, i.e., they would hold for less tuples. In Figure~\ref{fig:binning_methods} we can see that the cell coverage achieved by \subtab{} is much higher than the other baselines, and that this score moderately decreases for all three with the increase in number of bins. Indeed,  if association rules hold for less tuples, a sub-table would need to cover more rules in order to describe the same amount of cells.

Next, recall that \emph{support}~\cite{agrawal1994fast} reflects the ratio of tuples to which an association rule applies. By default, we set the minimum threshold for support to be~0.1, i.e., we are interested only in rules that hold for at least 10\% of the tuples. In Figure~\ref{fig:support} we vary the support threshold, and observe it only leads to a minor decrease in cell coverage, for the three algorithms. Intuitively, if we increase the threshold, less rules pass it, but these rules are prominent, so the embedding is still likely to capture them. 

The \emph{confidence} of an association rule measures the strength of the connection between its parts, i.e., the ratio, among all tuples for which the left-hand-side holds, of tuples where the right-hand-side holds as well. We have also varied the confidence threshold for association rules (by default, it was~0.1) in Figure~\ref{fig:confidence}. The observed trends are similar to the varying of support.

These results show the robustness of our approach across different properties of the association rules against which they are evaluated. In particular, the ranking between algorithms, and the relative gap between their scores is preserved across settings.

\balance{}
\subsection{Findings}\label{sec: findings}
Results of these experiments highlight the quality of sub-tables computed by \subtab{}, and show that they exceed those of other interactive algorithms and are comparable even to algorithms that directly optimize our metrics or use time-consuming state-of-the-art embedding methods. They also show that, unlike the baselines, \subtab{} is suitable for an interactive setting.  

From a usability perspective, experiments with pre-recorded EDA sessions show that \subtab{} outperforms other baselines by more frequently including columns and rows that were later used. The user study confirms that, as compared to the baselines, our sub-tables more frequently help data analysts derive useful insights from the data and increase user satisfaction.

Finally, the results also indicate that our metrics of sub-table quality are sound and robust, and correlate with external means of evaluating sub-table quality.

%% file: sections/conc.tex
\section{Conclusion and Future Work}\label{sec:conc}
This paper presents \subtab{}, a framework for creating small, informative sub-tables of large data tables to facilitate the first step in data analytics:  data exploration. Given a larger table, \subtab{} creates a sub-table, with a small subset of rows of the table projected over a small subset of columns, that could be explored manually by the analyst. The rows and columns are chosen as representatives of prominent data patterns within and across columns in the input table. \subtab{} can also be used for query results, enabling the user to quickly understand them and determine subsequent queries.

There are several directions for future research. Our current work considers only single dataset as input.  
We could consider handling multiple datasets, as well as optimizing sub-table computation for operations over multiple tables such as joins. Another intriguing future direction is creating sub-tables for other data science tasks, such as visualization and training ML models, both for the supervised and unsupervised setting. 

As we extend our approach to different tasks, one could explore different methods for table embedding, that could be performed offline or within a longer period. Finally, there are many other challenging variants of sub-table computation, e.g., computing sub-tables that meet certain fairness requirements with respect to the data they represent.